\documentclass[11pt,english]{article}
\usepackage[utf8]{inputenc}
\usepackage[english]{babel}
\usepackage{booktabs}
\usepackage{amsmath,amssymb,amsthm,amsfonts}
\usepackage{enumerate}
\usepackage{tikz,float,xcolor}
\usepackage[authoryear,round]{natbib}
\usepackage{hyperref}
\usepackage[a4paper,body={16cm,23.7cm}]{geometry}
\usepackage{typed-checklist}

\hypersetup{colorlinks=true,linkcolor=cyan,citecolor=cyan}

\newtheorem{theorem}{Theorem}
\newtheorem{lemma}{Lemma}

\newtheorem{remark}{Remark}
\newtheorem{example}{Example}

\DeclareMathOperator{\NM}{NM}

\begin{document}
	
	\title{\textbf{THE MUSEUM PASS PROBLEM WITH CONSORTIA}\thanks{Juan Carlos Gon\c{c}alves-Dosantos acknowledges the grant PID2021-12403030NB-C31 funded by MCIN/AEI/10.13039/501100011033 and by ``ERDF A way of making Europe/EU''. Ricardo Mart\'inez acknowledges the R\&D\&I project grant PID2023-147391NB-I00 funded by MCIN AEI/10.13039/501100011033 and by ``ERDF A way of making Europe/EU''. Joaqu\'in S\'anchez-Soriano acknowledges the grant PID2022-137211NB-I00 funded by MCIN/AEI/10.13039/501100011033 and by ``ERDF A way of making Europe/EU'' and from the Generalitat Valenciana under project PROMETEO/2021/063.}}
	\author{Juan Carlos Gon\c{c}alves-Dosantos\thanks{email: jgoncalves@umh.es} \\ Universidad Miguel Hern\'andez de Elche \and Ricardo Mart\'inez \\ Universidad de Granada \and Joaqu\'in S\'anchez-Soriano \\ Universidad Miguel Hern\'andez de Elche}
	\date{}
	\maketitle
	\begin{abstract}
		In this paper, we extend the museum pass problem to incorporate the market structure. To be more precise, we consider that museums are organized into several pass programs or consortia. Within this framework, we propose four allocation mechanisms based on the market structure and the principles of proportionality and egalitarianism. All these mechanisms satisfy different reasonable properties related to fairness and stability which serve to axiomatically characterize them.
	\end{abstract}

	\newpage
	\section{Introduction}
	
	\cite{Ginsburgh01} introduced the intriguing museum pass problem. The crux of this problem lies in how to fairly distribute the revenue generated from selling museum passes that grant access to a group of participating museums.\footnote{\cite{CasasMendez11} provides a good survey of the literature on the museum pass problem.} This issue extends beyond this particular application (see \cite{Algaba19}) and finds relevance in various real-life contexts where bundling products into packages proves more lucrative than selling them individually (\cite{Adams76}). In recent years, digital streaming platforms such as Spotify or Kindle have experienced significant growth potential. These platforms typically offer users a catalog of services at a fixed subscription price, allowing customers to choose the content or service they wish to consume. The challenge arises when these platforms need to compensate content producers (artists or writers) based on the consumption of each offered service. Previous research has examined the museum pass problem, including works by \cite{Martinez23}, \cite{Bergantinos15}, \cite{CasasMendez11}, \cite{Estevez-Fernandez12}, \cite{Wang11},  and \cite{Ginsburgh03}. However, these studies assume the existence of a single global pass in which all museums participate, excluding other intermediate self-organizing structures.\footnote{\cite{Bergantinos15} consider that pass holders may use the general pass or purchase tickets for individual museums, but intermediate structures are obviated.} Consider a city that offers tourists a pass granting access to all its monuments and museums. Typically, some of these monuments (such as cathedrals or monasteries) are owned by the local Church, while public museums coexist with private ones. Although the church or a consortium of private museums may participate in the city pass, they sometimes also offer their own passes restricted to the monuments and museums they manage. This more complex organizations must be considered, not only to allocate revenue from the sub-programs, but also to distribute the revenues obtained from selling the global pass. The novelty of this paper lies in incorporating intermediate market structures into the design of allocation mechanisms for the museum pass problem. The relationship between our model and findings and several of the previously mentioned works is elaborated upon in the final section.
	
	In our model, a \emph{problem} is described by five elements: the set of \emph{museums} that may participate in various programs, the set of \emph{consortia} representing the market structure, the sets of \emph{pass holders} who purchase each of the available passes, the list of \emph{pass prices} granting access to the museums involved in each program, and the list of \emph{consumption matrices} that indicate the museums visited by each pass holder. A \emph{rule} is a method to distribute among the museums the revenue obtained from selling all the passes.
	
	For our analysis, we adopt the axiomatic approach, which has a well-established tradition in economics literature dating back to \cite{Arrow51}. Instead of directly choosing from existing rules or alternatives, this methodology advocates selecting rules based on the axioms (or properties) they satisfy. In particular, the axioms we analyze in this paper implements different notions of fairness and stability. In the first group we include: \emph{dummy} (museums without visitors do not receive anything), \emph{symmetry within consortia} and \emph{symmetry between consortia}, which guarantee equitable treatment for museums and consortia with symmetric features. With regard stability, we study the property of \emph{composition} that states that the revenue-sharing process can be conducted in multiple stages without impacting the final allocation. In addition, we consider three axioms (\emph{splitting-proofness of museums}, \emph{splitting-proofness of consortia}, and \emph{the consortia property}) that guarantee that no museum or consortium may alter the allocation of other museums or consortia by artificially splitting their offer (segregating, for example, the main nave of a cathedral from its adjacent tower, or detaching different buildings of a same museum).
	
	With regard to the rules, the ones we propose are based on four key principles for revenue distribution. First, the consortium structure of museums should be reflected in the distribution process. We do that by considering two-stage mechanisms, where revenue is initially distributed among consortia and subsequently among museums within each consortium.  Second, only museums visited by a pass holder should receive a portion of the price paid for that pass, as they are the ones that attracted the pass holder's interest. The final two principles are equality and proportionality, which are widely used in practice and generally accepted as fair allocation methods. Thus, in the \emph{egalitarian-egalitarian rule} the pass of each pass holder is firstly split among the visited consortia, and in the second stage, each consortium's allocation is further divided among the museums that comprise it. In contrast, in the \emph{proportional-proportional rule} the the pass of each pass holder is initially divided among the consortia the pass holder visits, in proportion to the consortium prices, and then, within each consortium, this share is further proportionally distributed among the visited museums based on their individual prices. The other two allocation methods we propose, the \emph{egalitarian-proportional} and \emph{proportional-egalitarian rules} are crossed combinations of the two previous ones, mixing egalitarianism and proportionality. 
	
	In our main results, we establish normative foundations for the four aforementioned rules. We demonstrate that, among all possible revenue-distribution mechanisms, the egalitarian-egalitarian rule is the unique method that satisfies composition, dummy, and both symmetries (Theorem \ref{thm2}). Interestingly, we find out that if we replace fairness (symmetries) with non-manipulability (splitting-proofness and consortia property), the proportional-proportional rule is characterized instead (Theorem \ref{thm3}). Given these results, one might wonder whether fairness and non-manipulability are compatible. We show that the answer is affirmative (Theorems \ref{thm4} and \ref{thm5}).
	
	The rest of the paper is organized as follows. In section 2 we present the model and introduce the main rules we analyze. In Section 3 we propose several axioms that are suitable for this setting. In Section 4 we present our characterization results. Finally, in Section 5 we conclude with some final remarks. 
	
	\section{The consortia model}
	
	Let $\mathbb{M}$ represent the set of all potential museums and let $\mathcal{M}$ be the set of all finite (non-empty) subsets of $\mathbb{M}$. Now, let $\mathbb{N}$ represent the set of all potential buyers of a museum pass. Let $\mathcal{N}$ be the set of all finite (non-empty) subsets of $\mathbb{N}$. We denote by $\mathcal{P}^M$ the set of all possible partitions of $M$ (for any $P=\{P_1,\dots, P_s\} \in \mathcal{P}^M$ and any $P_k,P_l\in P$ we have $P_k\cap P_l=\varnothing$ and $\bigcup_{k=1}^{s}P_k=M$). A \textbf{museum pass problem with (a priori) consortia}, or simply a \textbf{problem}, is a 5-tuple $D=\left(M,P,N,\pi,C\right)$, where:\footnote{See Example \ref{example1} for an illustration of all the elements of the problem.}
	\begin{itemize}
		\item $M=\{1,\ldots,m\} \in \mathcal{M}$ is the set of \textbf{museums}.
		\item $P=\{P_1,\dots, P_s\} \in \mathcal{P}^M$ is the set of \textbf{consortia} that represent the market structure. 
		\item $N=\{N^{-m},\ldots,N^{-1},N^0,N^1,\ldots,N^s\} \in \mathcal{N}$ is the set of \textbf{pass holders}. Negative superindices refer to individual passes, zero superindex refers to the general pass, while positive superindices refer to consortium passes. Thus, for each $i \in M$, $N^{-i}$ indicates the pass holders that buy (and visit) the individual pass for museum $i$. Similarly, for each consortium $P_t \in P$, $N^{t}$ specifies the set of pass holders that buy the pass to visit all or some museums in the consortium $P_t$. Finally, $N^0$ indicates the pass holder that purchase the pass the allows the entrance to any museum in $M$. Notice that, for a given $\hat{N} \in N$, it may occur that $\hat{N} = \varnothing$ if no consumer buys the pass to enter to the corresponding museum or group of museums. For the sake of simplicity, if a consortium is formed by just one museum ($P_t=\{i\}$), we impose that $N^{-i}=\varnothing$. By doing this we avoid considering singletons twice. Moreover, we assume that the same person cannot buy two different passes; if this were the case, we would consider as many different pass holders as passes he or she had purchased.\footnote{We assume that each pass holder only belongs to a unique $\hat{N} \in N$, that is, for any pair $\hat{N},\overline{N} \in N$, $\hat{N} \cap \overline{N} = \varnothing$.}
		\item $\pi=\{\pi^{-m}, \ldots, \pi^{-1},\pi^0,\pi^1,\ldots,\pi^s\} $ is the set of \textbf{pass prices} where $\pi^{l}\in\mathbb{R}_{++}$ for all $l\in\{-m,\dots,-1,0,1,\dots,s\}$. Using the same convention as before, the first elements $\pi^{-m}, \ldots, \pi^{-1}$ indicate the pass price of each individual museum, $\pi^0$ is the price of the general pass, while $\pi^1,\ldots,\pi^s$ indicate the prices to visit the consortia. Again, for the sake of simplicity, if a consortium is formed by just one museum ($P_t=\{i\}$), we impose that $\pi^{-i}=\pi^t$.
		\item $C=\{C^{-m},\ldots,C^{-1},C^0,C^1,\ldots,C^s\}$ is the list of \textbf{consumption matrices}. The rows of all these matrices are labeled with the name of the museums and the columns with the name of the pass holders. Furthermore, all of these matrices are binary, where a $1$ in cell $ia$ means that museum $i$ has been visited by pass holder $a$, and $0$ otherwise. The first $m$ matrices $\{C^{-m},\ldots,C^{-1}\}$ refer to the consumption of individual museums. Each $C^{-i}$ has one row (for museum $i$) and $|N^{-i}|$ columns (one for each pass holder that buys that individual pass). The matrix $C^0$ refers to the consumption of the global pass. It has $m$ rows (one for each museum) and $|N^0|$ columns (one for each pass holder that buys the global pass). The last matrices $\{C^1,\ldots,C^s\}$ refer to the consumption of consortia. Each $C^{t}$ has $|P^t|$ rows (one for each museum in the consortium $P^t$) and $|N^{t}|$ columns (one for each pass holder that buys the corresponding consortium pass). For each $\hat{C} \in C$,
		$$
		\hat{C}_{ia}=
		\begin{cases}
			1 & \text{if museum $i$ has been visited by pass holder $a$} \\
			0 & \text{otherwise}
		\end{cases}
		$$
		For the sake of convention, if a pass has not been purchased by any pass holder ($\hat{N}=\varnothing$), the associated consumption matrix would have zero columns; in such a case we write $\hat{C}=\varnothing$. As in \cite{Ginsburgh03} and \cite{Bergantinos15}, we assume that any pass holder visits, at least, one museum.
	\end{itemize}
	
	The domain of all possible problems $(M,P,N,\pi,C)$ is denoted by $\mathcal{D}$. The total revenue obtained from selling the passes is given by
	$$
	E = \sum_{i=1}^m |N^{-i}|\pi^{-i} + |N^0|\pi^0 + \sum_{t=1}^s |N^t|\pi^{t},
	$$
	where $\sum_{i=1}^m |N^{-i}|\pi^{-i}$ is the revenue from individual passes ($|N^{-i}|\pi^{-i}$ corresponds to the revenue generated by museum $i$), $|N^0|\pi^0$ is the revenue from the global pass, and $\sum_{t=1}^s |N^t|\pi^{t}$ is the revenue from consortium passes ($|N^t|\pi^{t}$ corresponds to the revenue generated by consortium $t$).
	
	Now, we introduce some notation which we will use in the rest of the paper:
	\begin{itemize}
		\item For each $i \in M$ and each $\sigma \in\{-m,\dots,-1,0,1,\dots,s\}$, $N^\sigma_i =\{a \in N^\sigma: C^\sigma_{ia}=1\}$ is the set of pass holders in $N^\sigma$ that visit museum $i$. 
		\item For each $i \in M$, $P^{(i)}$ is the consortium $i$ belongs to, i.e., $P^{(i)} \in P$ is the unique consortium such that $i \in P^{(i)}$. Similarly, $N^{(i)}$ and $C^{(i)}$ are the set of pass holders and consumption matrix associated with consortium $P^{(i)}$.
		\item For each $a \in N$, $M_a =\{i \in M: \hat{C}_{ia}=1 \text{ for some } \hat{C} \in C\}$ is the set of museum visited by pass holder $a$. 
		\item For each $a \in N^0$, $K^0_{a}=\{k\in \{1,\dots,s\}: C^0_{ia}=1 \text{ for some } i \in P_k \}$ is the set of consortia visited (with a general pass) by $a$.
		\item For each $\hat{C} \in C$ and each $i \in M$, we define $\hat{C}_{i \cdot}=\sum_{a \in \hat{N}} \hat{C}_{ia}$ as the number of visitors of $i$ in the consumption matrix $\hat{C}$.
		\item Let $\mathcal{D}^0 \subset \mathcal{D}$ be the subclass of problems in which pass holders only buy general passes (i.e. $(M,P,N,\pi,C) \in \mathcal{D}^0$ iff $\hat{N} = \varnothing$ for any $\hat{N} \in N \backslash N^0$).
		\item For each $P_t \in P$, let $\mathcal{D}^t \subset \mathcal{D}$ be the subclass of problems in which pass holders only buy the pass to access museums in the consortium $t$ (i.e. $(M,P,N,\pi,C) \in \mathcal{D}^t$ iff $\hat{N} = \varnothing$ for any $\hat{N} \in N \backslash N^t$).
		\item For each $i \in M$, let $\mathcal{D}^{-i} \subset \mathcal{D}$ be the subclass of problems in which pass holders only buy passes to access museum $i$ (i.e. $(M,P,N,\pi,C) \in \mathcal{D}^{-i}$ iff $\hat{N} = \varnothing$ for any $\hat{N} \in N \backslash N^{-i}$).
	\end{itemize}
	
	A \textbf{rule} is a mechanism to distribute the generated revenue among the museums. Formally, it is a mapping $R:\mathcal{D}\longrightarrow \mathbb{R}^m_{+}$ that, for each problem $D \in \mathcal{D}$, determines an allocation $R(D) \in \mathbb{R}^m_+$ such that
	$$
	\sum_{i\in M}R_i(D)=E.
	$$
	
	\section{Revenue allocation mechanisms}
	
	In this section we present four natural rules that are suitable for this setting which propose different alternatives to distribute $E$. As the aggregate revenue is sum of the revenues from the general pass, the consortium passes, and the revenue from each individual museum, it is nature for the rules to operate in a similar manner and to divide the total revenue into the three contributing levels. In particular, the rules we propose are designed on four principles when establishing the distribution of revenues. The first is that the consortium structure of museums should be reflected in the distribution process, therefore the rules should be designed in stages. The second principle is that only museums visited by a pass holder should receive part of the price paid for that pass, since they are the ones that aroused the interest of the pass holder. The last two principles are equality and proportionality when allocating the revenues. The application of one or the other of these last two principles will depend on the relationship between the museums or between the consortia, and the relevance the central planner gives to the prices of the passes. The next four proposals divide this distribution process into two stages. In the first stage, the pass of each pass holder is allocated among the visited consortia, and in the second stage, each consortium’s allocation is further divided among the museums that comprise it.
	
	The first rule we introduce is the \emph{egalitarian-egalitarian rule}. In this rule, at the first stage the entrance fee $\pi^0$ is equally divided among the consortia the pass holder visits. At the second stage this share is further equally distributed among the visited museums within each consortium. The pass of the consortium, $\pi^{(i)}$, is also equally allocated among the visited museums, and the revenue associated with the sale of individual passes, $|N^{-i}|\pi^{-i}$, is assigned to the corresponding museum.
	
	\textbf{Egalitarian-egalitarian rule}: For each $D \in \mathcal{D}$ and each $i\in M$, 
	$$
	R^{EE}_{i}(D)=\sum_{a \in N^0_i} \frac{1}{|M_a \cap P^{(i)}|} \frac{1}{|K^0_a|} \pi^0+\sum_{a \in N^{(i)}_i} \frac{1}{|M_a|}\pi^{(i)}+|N^{-i}|\pi^{-i}
	$$
	
	In the next rule, at the first stage the general pass $\pi^0$ is initially divided among the consortia the pass holder visits, in proportion to the consortium prices. Then, within each consortium, this share is further proportionally distributed among the visited museums based on their individual prices. The consortium pass $\pi^{(i)}$ is also proportionally allotted among the visited museums, using the individual prices as a reference.
	
	\textbf{Proportional-proportional rule}: For each $D \in \mathcal{D}$ and each $i\in M$, 
	$$
	R^{PP}_{i}(D)=\sum_{a \in N^0_i} \frac{\pi^{-i}}{\sum_{j \in M_a \cap P^{(i)}} \pi^{-j}} \frac{\pi^{(i)}}{\sum_{t \in K^0_a}\pi^t} \pi^0+\sum_{a \in N^{(i)}_i} \frac{\pi^{-i}}{\sum_{j\in M_a}\pi^{-j}} \pi^{(i)}+|N^{-i}|\pi^{-i}
	$$
	
	The following two rules are mixtures of the previous ones. As for the \emph{proportional-egalitarian rule}, the general entrance $\pi^0$ paid by each pass holder is firstly divided among the consortia she visits in proportion to their consortium prices, and secondly, this share is then equally distributed among the visited museums within each consortium. The \emph{egalitarian-proportional rule} is the reverse process. The general pass is firstly divided uniformly among the visited consortia, and then in proportion to the individual prices within each consortium. 
	
	\textbf{Proportional-egalitarian rule}: For each $D \in \mathcal{D}$ and each $i\in M$, 
	$$
	R^{PE}_{i}(D)=\sum_{a \in N^0_i} \frac{1}{|M_a \cap P^{(i)}|} \frac{\pi^{(i)}}{\sum_{t \in K^0_a}\pi^t} \pi^0+\sum_{a \in N^{(i)}_i} \frac{1}{|M_a|}\pi^{(i)}+|N^{-i}|\pi^{-i}
	$$
	
	\textbf{Egalitarian-proportional rule}: For each $D \in \mathcal{D}$ and each $i\in M$, 
	$$
	R^{EP}_{i}(D)=\sum_{a \in N^0_i} \frac{\pi^{-i}}{\sum_{j \in M_a \cap P^{(i)}} \pi^{-j}} \frac{1}{|K^0_a|} \pi^0+ \sum_{a \in N^{(i)}_i} \frac{\pi^{-i}}{\sum_{j\in M_a}\pi^{-j}} \pi^{(i)} +|N^{-i}|\pi^{-i}
	$$
	
	In the following example we explain in detail the elements of the model and illustrate how the aforementioned rules operate.
	
	\begin{example}\label{example1}
		Consider the problem $D \in \mathcal{D}$ with three museums, $M=\{1,2,3\}$, organized into two consortia $P=\{\{1,2\},\{3\}\}$, and ten pass holders, $N=\{N^{-3},N^{-2},N^{-1},N^0,N^1,N^2\}$, where
		$$
		N^{-3}=\varnothing, \quad N^{-2}=\{1,2,3\}, \quad N^{-1}=\{4\}, \quad N^0=\{5,6\}, \quad N^1=\{7,8\}, \quad N^2=\{9,10\}.
		$$
		The pass prices are 
		$$
		\pi^{-3}=3, \quad \pi^{-2}=2, \quad \pi^{-1}=1, \quad \pi^0=4, \quad \pi^1=2, \quad \pi^2=3,
		$$
		and the visits are described by the following consumption matrices
		$$
		C^{-3}=\varnothing, \quad C^{-2}= \left( \begin{array}{ccc} 1 & 1 & 1  \end{array} \right),  \quad C^{-1}= \left( \begin{array}{c} 1  \end{array} \right),
		$$
		$$
		C^0= \left( \begin{array}{cc} 1 & 0 \\ 1 & 0 \\ 1 & 1  \end{array} \right)
		$$
		$$
		C^{1}= \left( \begin{array}{cc} 1 & 1 \\ 0 & 1 \end{array} \right), \quad C^{2}= \left( \begin{array}{cc} 1 & 1  \end{array} \right)
		$$
		Therefore,
		$$
		E = [1 \cdot 1 + 3 \cdot 2 + 0 \cdot 3] + 2 \cdot 4 + [2 \cdot 2 + 2 \cdot 3] = 25
		$$
		Several comments are in order. First, note that the consortium $P_2$ consists of just one pass holder. As outlined in the model's setup, for consistency, we assume $N^{-3}=\varnothing$ (to avoid double-counting of visitors) and $\pi^{-3}=\pi^2$ (prices remain the same whether 3 is considered as a consortium of a singleton or as an individual museum). Secondly, $N^{-2}=\{1,2,3\}$ indicates that pass holders 1, 2, and 3 have bought tickets to access museum 2 at a price of $\pi^{-2}=2$. Similarly, pass holders 6 and 7 have purchased passes granting entry to all museums at a cost of $\pi^0=4$, and pass holders in $N^1=\{7,8\}$ have acquired combined entrance tickets to the consortium $P_1$, which includes museums 1 and 2. Thirdly, regarding the consumption matrices, $C^{-3}=\varnothing$ for consistency. All entries in $C^{-1}$ and $C^{-2}$ are set to one, based on the assumption that any pass holder visits at least one of the museums included in their pass.\footnote{This is a standard assumption in the literature on the museum pass problem. In the final remarks, we provide a more detailed discussion on its implications and possible relaxation.}
		
		Below is a detailed look at how each of the four rules introduced previously works.
		\begin{itemize}
			\item Egalitarian-egalitarian rule. Regarding the global pass, the entrance fee of pass holder $5$ is initially divided equally among the consortia she visits, allocating $\frac{\pi^0}{2}$ to both $P_1$ and $P_2$. Within each consortium, this share is further equally distributed among the visited museums. Consequently, museum 1 receives $\frac{1}{2} \cdot \frac{\pi^0}{2}$, museum 2 receives $\frac{1}{2} \cdot \frac{\pi^0}{2}$, and museum 3 receives $\frac{\pi^0}{2}$. A similar procedure is applied to the pass paid by pass holder 6, as well as the revenue generated from selling the passes of consortium $P_1$. The amount paid by pass holder 7 is allocated to museum 1, while the price $\pi^1=2$ paid by pass holder 8 is divided between museums 1 and 2. Since the second consortium, $P_2$, consists of only one museum, it receives all the generated revenue. Regarding individual passes, each museum obtains its respective profit.
			\begin{align*}
				R^{EE}_1(D) &= \left[ \frac{1}{2} \cdot \frac{1}{2} \cdot 4+0 \cdot 4 \right] + \left[ 1 \cdot 2 + \frac{1}{2} \cdot 2 \right] + \left[ 1 \cdot 1 \right] = 5 \\
				R^{EE}_2(D) &= \left[ \frac{1}{2} \cdot \frac{1}{2} \cdot 4+0 \cdot 4 \right] + \left[ 0 \cdot 2 + \frac{1}{2} \cdot 2 \right] + \left[ 3 \cdot 2 \right] = 8 \\
				R^{EE}_3(D) &=\left[ \frac{1}{2} \cdot 1 \cdot 4+ 1 \cdot 4 \right] + \left[ 2 \cdot 3  \right] + \left[ 0 \cdot 3 \right] = 12
			\end{align*}
			\item Proportional-proportional rule. The ticket paid by pass holder 5 is initially divided among the consortia she visits in proportion to their consortium prices, allocating $\frac{2}{5}\pi^0$ to $P_1$ and $\frac{3}{5} \pi^0$ $P_2$. Within each consortium, this share is further distributed among the visited museums in proportion to their corresponding individual prices. Consequently, museum 1 receives $\frac{1}{3} \cdot \frac{2}{5} \cdot \pi^0$, museum 2 receives $\frac{2}{3} \cdot \frac{2}{5} \cdot \pi^0$, and museum 3 receives $\frac{3}{5} \cdot \pi^0$. The process pass holder 6 works similarly. Regarding the revenue generated by consortium $P_1$, the entrance fee paid by pass holder 7 goes entirely to museum 1, while the price $\pi^1=2$ paid by pass holder 8 is distributed between museums 1 and 2 in proportion to the prices of their individual passes ($\frac{1}{3}$ for museum 1 and $\frac{2}{3}$ for museum 2). The computation of the remaining revenue allocation is straightforward. Thus, the overall revenue distribution according to the proportional-proportional rule is as follows:
			\begin{align*}
				R^{PP}_1(D) &= \left[ \frac{1}{3} \cdot \frac{2}{5} \cdot 4+0 \cdot 4 \right] + \left[ 1 \cdot 2 + \frac{1}{3} \cdot 2 \right] + \left[ 1 \cdot 1 \right] = \frac{21}{5} \\
				R^{PP}_2(D) &= \left[ \frac{2}{3} \cdot \frac{2}{5} \cdot 4+0 \cdot 4 \right] + \left[ 0 \cdot 2 + \frac{2}{3} \cdot 2 \right] + \left[ 3 \cdot 2 \right] =  \frac{42}{5} \\
				R^{PP}_3(D) &=\left[ 1 \cdot \frac{3}{5} \cdot 4+ 1 \cdot 4 \right] + \left[ 2 \cdot 3  \right] + \left[ 0 \cdot 3 \right] = \frac{62}{5}
			\end{align*}
			\item Proportional-egalitarian rule. The entrance by pass holder 5 is firstly divided among the consortia she visits in proportion to their consortium prices, allocating $\frac{2}{5}\pi^0$ to $P_1$ and $\frac{3}{5} \pi^0$ $P_2$. Within each consortium, this share is further equally distributed among the visited museums. Then, museum 1 receives $\frac{1}{2} \cdot \frac{2}{5} \cdot \pi^0$, museum 2 receives $\frac{1}{2} \cdot \frac{2}{5} \cdot \pi^0$, and museum 3 receives $\frac{3}{5} \cdot \pi^0$. A similar procedure is applied to the pass paid by pass holder 6. The distribution of revenues generated by the consortia and individual passes works as in the egalitarian-egalitarian rule.
			\begin{align*}
				R^{PE}_1(D) &= \left[ \frac{1}{2} \cdot \frac{2}{5} \cdot 4+0 \cdot 4 \right] + \left[ 1 \cdot 2 + \frac{1}{2} \cdot 2 \right] + \left[ 1 \cdot 1 \right] = \frac{24}{5} \\
				R^{PE}_2(D) &= \left[ \frac{1}{2} \cdot \frac{2}{5} \cdot 4+0 \cdot 4 \right] + \left[ 0 \cdot 2 + \frac{1}{2} \cdot 2 \right] + \left[ 3 \cdot 2 \right] = \frac{39}{5} \\
				R^{PE}_3(D) &=\left[ 1 \cdot \frac{3}{5} \cdot 4+ 1 \cdot 4 \right] + \left[ 2 \cdot 3  \right] + \left[ 0 \cdot 3 \right] = \frac{62}{5}
			\end{align*}
			\item Egalitarian-proportional rule. Regarding the global pass, the entrance fee of pass holder 5 is initially divided equally among the consortia she visits, allocating $\frac{\pi^0}{2}$ to both $P_1$ and $P_2$. Within each consortium, this share is further distributed among the visited museums in proportion to their corresponding individual prices. Thus, museum 1 receives $\frac{1}{3} \cdot \frac{1}{2} \cdot \pi^0$, museum 2 receives $\frac{2}{3} \cdot \frac{1}{2} \cdot \pi^0$, and museum 3 receives $\frac{1}{2} \cdot \pi^0$. The process pass holder 6 works similarly. The distribution of revenues generated by the consortia and individual passes works as in the proportional-proportional rule.
			\begin{align*}
				R^{EP}_1(D) &= \left[ \frac{1}{3} \cdot \frac{1}{2} \cdot 4+0 \cdot 4 \right] + \left[ 1 \cdot 2 + \frac{1}{3} \cdot 2 \right] + \left[ 1 \cdot 1 \right] = \frac{13}{3} \\
				R^{EP}_2(D) &= \left[ \frac{2}{3} \cdot \frac{1}{2} \cdot 4+0 \cdot 4 \right] + \left[ 0 \cdot 2 + \frac{2}{3} \cdot 2 \right] + \left[ 3 \cdot 2 \right] =  \frac{26}{3} \\
				R^{EP}_3(D) &=\left[ 1 \cdot \frac{1}{2} \cdot 4+ 1 \cdot 4 \right] + \left[ 2 \cdot 3  \right] + \left[ 0 \cdot 3 \right] = \frac{36}{3}
			\end{align*}
		\end{itemize}

	\end{example}

	\section{Normative framework: Axioms}
	In general, there is no single criterion that determines which allocation rule to use. For this reason, in this section we introduce a set of axioms or properties that are reasonable for the problem framework being addressed. Depending on the axioms that are considered relevant in each specific situation, one rule or another will be selected. The first property says that the allocation is independent of the timing. Imagine the following two alternatives. One, we distribute the revenue every semester, considering $N_1$ and $N_2$ two disjoint groups of pass holders in each period. And two, we solve the problem once a year, considering the revenue generated by the whole group of pass holders, $N_1 \cup N_2$. In both cases the allocation must be the same. Similar principles have been applied, among others, by \cite{Martinez23} \cite{Martinez22d}, \cite{Martinez21} and \cite{Bergantinos15} in analogous contexts.
	
	\textbf{Composition.}  For each pair $(M,P,N_1,\pi,C_1), (M,P,N_2,\pi,C_2) \in \mathcal{D}$, with $N_1\cap N_2=\phi$
	$$
	R(M,P,N_1 \cup N_2,\pi,(C_1,C_2))=R(M,P,N_1,\pi,C_1)+R(M,P,N_2,\pi,C_2),
	$$
	where $(C_1,C_2)$ are the matrices of properly concatenating $C_1$ and $C_2$.
	
	The following two axioms embody the fundamental fairness principle that equals should be treated equally. The first axiom asserts that if two museums belong to the same consortium and have been visited by the same pass holders, their awards must also be identical. The second axiom extends a similar requirement to consortia that are considered equal. Although in distinct contexts, both properties are standard requirements in the literature (e.g. \cite{AlonsoMeijide20, AlonsoMeijide15}, \cite{Calvo10}).
	
	\textbf{Symmetry within consortia.}  For each $D \in \mathcal{D}$ and each pair $i,j\in M$ such that $P^{(i)}=P^{(j)}$, if the three following conditions hold
	\begin{enumerate}
		\item[(i)]  $C^0_{ia} = C^0_{ja}$ for any $a \in N^0$. 
		\item[(ii)] $C^{(i)}_{ia}=C^{(j)}_{ja}$ for any $a \in N^{(i)}=N^{(j)}$. 
		\item[(iii)]  $\pi^{-i}|N^{-i}|=\pi^{-j}|N^{-j}|$.
	\end{enumerate}
	Then, $R_i(D)=R_j(D)$.
	
	Consider two museums $i,j \in M$, both belonging to the same consortium. Condition (i) says that every pass holder who purchases the general pass visits both museums.\footnote{Notice that, if no one buys the general pass, i.e. $C^0=\varnothing$, the condition is vacuously satisfied.} Analogously, condition (ii) requires that anyone who buys the consortium pass accesses both $i$ and $j$. Finally, (iii) stipulates that at individual level, both $i$ and $j$ generate the same individual revenue. 
	
	The next property adapts the previous principle to symmetric consortia. 
	
	\textbf{Symmetry between consortia.}  For each $D \in \mathcal{D}$ and each pair $P_r,P_t\in P$, if the following three conditions hold 
	\begin{enumerate}
		\item[(i)] $\sum_{i \in P_r} C^0_{i\cdot} > 0 \Leftrightarrow \sum_{i \in P_t} C^0_{i\cdot} > 0$.
		\item[(ii)] $\pi^r|N^r|=\pi^t|N^t|$ 
		\item[(iii)] $\sum_{i \in P_r} \pi^{-i}|N^{-i}|= \sum_{j \in P_t} \pi^{-j}|N^{-j}|$
	\end{enumerate}
	Then, $\sum_{i \in P_r} R_i(D)= \sum_{j \in P_t} R_j(D)$.
	
	The next requirement relies on a well-established principle: any museum that does not contribute to revenue should be excluded from distribution. We say that a museum  $i \in M$ is \emph{dummy} if it has no visitors, either because they do not buy any pass that includes access to $i$, or because even with a pass that grants access to $i$, they do not visit it. The axioms states that the allocation of $i$ should be zero. Let $D \in \mathcal{D}$, we denote by $\NM(D)$ the set of dummy museums in problem $D$. More specifically,
	$$
	\NM(D) = \{i \in M : \text{ for each } \sigma \in \{-i,0,(i)\}, \text{ either } C^\sigma=\varnothing \text{ or } C^\sigma_{ia}=0 \; \forall a \in N^\sigma \}
	$$
	
	\textbf{Dummy.}  For each $D \in \mathcal{D}$ and each $i \in M$, if $i \in \NM(D)$ then $R_i(D)=0.$
	
	The following two properties establish stability criteria in a specific sense. If a museum chooses to divide into multiple entities, the other museums remain unaffected. For instance, imagine that Louvre Museum artificially splits its whole collection (between painting and sculpture, for example) pretending to act as two different museums, and this choice does not alter the revenues. The next requirement states that Louvre’s strategy does not impact the allocation of resources among the other museums. This same principles have already been applied to similar setting. See, for example, \cite{Slikker23}, \cite{Knudsen12}, \cite{Moulin07}, or \cite{Ju03a}.
	
	\textbf{Splitting-proofness of museums.}  Let $(M,P,N,\pi,C) \in \mathcal{D}$ and $i\in M$. Consider $(M',P',N',\pi',C') \in \mathcal{D}$ such that
	\begin{itemize}
		\item $M'=\left(M\backslash\{i\}\right)\cup\{i_1,\dots,i_r\}$.
		\item $P'=\left(P\backslash P^{(i)}\right)\cup \left(\left(P^{(i)}\backslash\{i\}\right)\cup\{i_1,\dots,i_r\}\right)$.
		\item $N'=\left(N\backslash N^{-i}\right)\cup\{N^{-i_1},\dots,N^{-i_r}\}$ where $N^{-i_h}=N^{-i}$ for all $h\in\{1,\dots,r\}$.
		\item $\pi'=\left(\pi\backslash\pi^{-i}\right)\cup\{\pi^{-i_1},\dots,\pi^{-i_r}\}$ where $\sum_{h=1}^{r}\pi^{-i_h}=\pi^{-i}$.
		\item $C'=\left(C\backslash\left(C^0 \cup C^{(i)}\cup C^{-i}\right)\right)\cup C'^0 \cup C'^{(i)}\cup\bigcup_{l=1}^{t}\{C^{-i_1},\dots,C^{-i_r}\}$ where $C'^0_{ja}=C^0_{ja}$ and $C'^0_{i_ha}=C^0_{ia}$, for all $a \in N^0$, $j \in M\backslash\left\{i\right\}$, $h \in \{1,\dots,r\}$; $C'^{(i)}_{ja}=C^{(i)}_{ja}$, $C'^{(i)}_{i_ha}=C^{(i)}_{ia}$ and $C^{-i_h}=C^{-i}$ for all $j\in P^{(i)}\backslash\{i\}$ and $h\in\{1,\dots,r\}$. 
	\end{itemize}
	Then, for each $j\in M\backslash\{i\}$,
	$$
	R_j(M,P,N,\pi,C)=R_j(M',P',N',\pi',C').
	$$
	
	The following property modifies the earlier principle to be applied to consortia instead to individual museums. 
	
	\textbf{Splitting-proofness of consortia.}  Let $(M,P,N,\pi,C) \in \mathcal{D}$ and $P_k\in P$ where $P_k=\{i_1,\dots,i_r\}$. Consider $P_{k^1},\dots,P_{k^t}$ such that $P_{k^l}=\{i^l_1,\dots,i^l_r\}$ for all $l\in\{1,\dots,t\}$; and $(M',P',N',\pi',C') \in \mathcal{D}$ be such that
	\begin{itemize}
		\item $M'=\left(M\backslash\{i_1,\dots,i_r\}\right)\cup\bigcup_{l=1}^{t}\{i^l_1,\dots,i^l_r\}$.
		\item $P'=\left(P\backslash P_k\right)\cup\bigcup_{l=1}^{t}P_{k^l}$.
		\item $N'=\left(N\backslash\left(N^k\cup\bigcup_{j=i_1}^{i_r}N^{-j}\right)\right)\cup\bigcup_{l=1}^{t}N^{k_l}\cup\bigcup_{l=1}^{t}\{N^{-i^l_1},\dots,N^{-i^l_r}\}$ where $N^{k_l}=N^{k}$ and $N^{-i^l_h}=N^{-i_h}$ for all $l\in\{1,\dots,t\}$, $h\in\{1,\dots,r\}$.
		\item $\pi'=\left(\pi\backslash\left(\pi^k\cup\bigcup_{j=i_1}^{i_r}\pi^{-j}\right)\right)\cup\bigcup_{l=1}^{t}\pi^{k_l}\cup\bigcup_{l=1}^{t}\{\pi^{-i^l_1},\dots,\pi^{-i^l_r}\}$ where $\sum_{l=1}^{t}\pi^{k_l}=\pi^{k}$ and $\sum_{l=1}^{t}\pi^{-i^l_h}=\pi^{-i_h}$ for all $h\in\{1,\dots,r\}$. 
		\item $C'=\left(C\backslash\left(C^0 \cup C^k\cup\bigcup_{j=i_1}^{i_r}C^{-j}\right)\right)\cup C'^0 \cup \bigcup_{l=1}^{t}C^{k_l}\cup\bigcup_{l=1}^{t}\{C^{-i^l_1},\dots,C^{-i^l_r}\}$ where $C'^0_{ja}=C^0_{ja}$, $C'^0_{i^l_ha}=C^0_{i_ha}$, for all $a \in N^0$, $j \in M \backslash\{i_1,\dots,i_r\}$, $h\in\{1,\dots,r\}$, $l\in\{1,\dots,t\}$; and $C^{k_l}=C^{k}$ and $C^{-i^l_h}=C^{-i_h}$ for all $l\in\{1,\dots,t\}$, $h\in\{1,\dots,r\}$. 
	\end{itemize}
	Then, for each $P_r \in P \backslash \{P_k\}$
	$$
	\sum_{j\in P_r}R_j(M,P,N,\pi,C)=\sum_{j\in P_r}R_j(M',P',N',\pi',C')
	$$
	
	The last properties also pertain to the stability of the allocation with respect to potential manipulations. In the previous two axioms, a museum (or a consortium) was split into several museums (or consortia). The manipulation referred to by the following property is an intermediate case, where a consortium intends to operate as a single museum. More specifically, consider the family of problems $\mathcal{D}^0$ (where pass holders only acquire the general pass). If each consortium acts as a single museum, the axiom requires that the allocation of this single museum coincides with the aggregated allotment of the consortium in the original problem.
	
	\textbf{Consortia consistency.} For each $(M,P,N,\pi,C) \in \mathcal{D}^0$ and each $P_k \in P$,
	$$
	\sum_{i\in P_k}R_i(M,P,N,\pi,C)=R_k\left(\{1,...,s\},\hat{P}^s,\hat{N},\hat{\pi},\hat{C}\right),
	$$
	where $\left(\{1,...,s\},P^s,\hat{N},\hat{\pi},\hat{C}\right)\in \mathcal{D}^0$ is such that $\hat{P}^s=\{\{1\},...,\{s\}\}$, $\hat{N}=\{\emptyset,\ldots,\emptyset,N^0,\emptyset,\ldots,\emptyset\}$, $\hat{\pi}=\{\hat{\pi}^{-s}, \ldots, \hat{\pi}^{-1},\pi^0,\hat{\pi}^1,\ldots,\hat{\pi}^s\} $ with $\hat{\pi}^{-t}=\hat{\pi}^{t}=\pi^{t}$ for any $t\in\{1,...,s\}$, and $\hat{C}=\{\emptyset,\ldots,\emptyset,\hat{C}^0,\emptyset,\ldots,\emptyset\}$ with, for any $t \in \{1,\ldots,s\}$ and $a \in \hat{N}^0$
	$$
	\hat{C}^0_{ta}=
	\begin{cases}
		1 & \text{if } \exists i \in P_t: C_{ia}=1 \\
		0 & \text{otherwise.}
	\end{cases}
	$$
	
	For each $D \in \mathcal{D}^0$, we will refer to $\left(\{1,...,s\},\hat{P}^s,\hat{N},\hat{\pi},\hat{C}\right)$ as the \emph{reduced problem} associated to $D$.
	
	Splitting-proofness of museums, splitting-proofness of consortia, and consortia consistency require the rule to be immune to specific (and distinct) types manipulations. Remarks \ref{remark2} and \ref{remark3} on the tightness of the characterization results prove that these three axioms are logically independent of one another.
	
	\section{Axiomatic characterization of the rules}
	
	In this section we present our main characterization results. For the sake of exposition and to avoid unnecessary repetition of reasoning, we first prove some technical lemmas, which are relegated to Appendix A.
	
	\begin{theorem}\label{thm2}
		A rule satisfies symmetry within consortia, symmetry between consortia, dummy, and composition if and only if it is the egalitarian-egalitarian rule.
		\begin{proof}
			We start by showing that the \emph{egalitarian-egalitarian rule} satisfies the axioms in the statement.
			\begin{itemize}
				\item Composition. Let $(M,P,N,\pi,C) , (M,P,N',\pi,C')  \in \mathcal{D}$ such that $N \cap N' = \emptyset$. Let $i \in M$. It follows that
				\begin{align*}
					R^{EE}_{i}(M,P,N\cup N',\pi\cup\pi,C\cup C') &= \sum_{a \in N^0_i\cup N'^0_i}\frac{1}{|M_a \cap P^{(i)}|} \frac{1}{|K^0_a|} \pi^0\\
					&+ \sum_{a \in N^{(i)}_i\cup N'^{(i)}_i} \frac{1}{|M_a|}\pi^{(i)}+|N^{-i}|\pi^{-i}+|N'^{-i}|\pi^{-i}\\
					&=\sum_{a \in N^0_i}\frac{1}{|M_a \cap P^{(i)}|} \frac{1}{|K^0_a|} \pi^0+\sum_{a \in N^{(i)}_i} \frac{1}{|M_a|}\pi^{(i)}+|N^{-i}|\pi^{-i}\\
					&+\sum_{a \in  N'^0_i}\frac{1}{|M_a \cap P^{(i)}|} \frac{1}{|K^0_a|} \pi^0+\sum_{a \in N'^{(i)}_i} \frac{1}{|M_a|}\pi^{(i)}+|N'^{-i}|\pi^{-i}\\
					&=R^{EE}_{i}(M,P,N,\pi,C) + R^{EE}_{i}(M,P,N',\pi,C') .
				\end{align*}
				\item Dummy. Let $D \in \mathcal{D}$ and $i \in \NM(D)$. Since $N^0_i=N^{(i)}_i=N^{-i}=\varnothing$, then
				$$
				R^{EE}_{i}(D) = \sum_{a \in N^0_i}\frac{1}{|M_a \cap P^{(i)}|} \frac{1}{|K^0_a|} \pi^0+\sum_{a \in N^{(i)}_i} \frac{1}{|M_a|}\pi^{(i)}+|N^{-i}|\pi^{-i}=0.
				$$
				\item Symmetry within consortia. Let $D \in \mathcal{D}$. Let $i,j\in M$ that satisfy the conditions in the definition of the property. Then
				\begin{align*}
					R^{EE}_{i}(D) &= \sum_{a \in N^0_i}\frac{1}{|M_a \cap P^{(i)}|} \frac{1}{|K^0_a|} \pi^0+\sum_{a \in N^{(i)}_i} \frac{1}{|M_a|}\pi^{(i)}+|N^{-i}|\pi^{-i}\\
					&=\sum_{a \in N^0_j}\frac{1}{|M_a \cap P^{(j)}|} \frac{1}{|K^0_a|} \pi^0+\sum_{a \in N^{(j)}_i} \frac{1}{|M_a|}\pi^{(j)}+|N^{-j}|\pi^{-j}\\
					&=R^{EE}_{j}(D). 
				\end{align*}
				\item Symmetry between consortia. Let $D \in \mathcal{D}$. Let $P_r,P_t\in P$ that satisfy the conditions in the definition of the property. Then
				\begin{align*}
					\sum_{i\in P_r} R^{EE}_{i}(D) &= \sum_{i\in P_r}\sum_{a \in N^0_i}\frac{1}{|M_a \cap P^{(i)}|} \frac{1}{|K^0_a|} \pi^0+\sum_{i\in P_r}\sum_{a \in N^{(i)}_i} \frac{1}{|M_a|}\pi^{(i)}+\sum_{i\in P_r}|N^{-i}|\pi^{-i} \\
					&= \sum_{a \in N^0:r\in K^0_a} \frac{1}{|K^0_a|} \pi^0+\pi^r|N^r|+\sum_{j\in P_t}|N^{-j}|\pi^{-j} \\
					&= \sum_{a \in N^0:t\in K^0_a} \frac{1}{|K^0_a|} \pi^0+\pi^t|N^t|+\sum_{j\in P_t}|N^{-j}|\pi^{-j} \\
					&= \sum_{j\in P_t}\sum_{a \in N^0_j}\frac{1}{|M_a \cap P^{(j)}|} \frac{1}{|K^0_a|} \pi^0+\sum_{j\in P_t}\sum_{a \in N^{(j)}_j} \frac{1}{|M_a|}\pi^{(j)}+\sum_{j\in P_r}|N^{-j}|\pi^{-j} \\
					&= \sum_{j\in P_t}  R^{EE}_{j}(D). 
				\end{align*}
			\end{itemize}
			Now, we prove the converse. Let $R$ be a rule that satisfies the properties in the statement, and let $(M,P,N,\pi,C) \in \mathcal{D}$. We divide the proof into several steps.
			\begin{itemize}
				\item[(i)] Let $i \in M$, and let $D \in \mathcal{D}^{-i}$. By Lemma \ref{lemma.3.}, $R_j(D) = 0 = R^{EE}_j(D)$ for any $j \in M\backslash \{i\}$, and then $R_i(D)=\pi^{-i} |N^{-i}|= R^{EE}_i(D)$ by definition of rule. Thus, $R$ coincides with the egalitarian-egalitarian rule in any subclass of problems $\mathcal{D}^{-i}$.
				\item[(ii)] Let $D^0=(M,P,\{a\},\pi,C^{(a)}) \in \mathcal{D}^0$ be a problem with just one pass holder purchasing the general pass, and let $i \in M$ with $i\in P_k$. Notice that, if $k \notin K^0_a$, then \emph{dummy} implies that $R_i(D^0)=0=R^{EE}_i(D^0)$. If $k \in K^0_a$, in application of \emph{symmetry between consortia}, all visited consortia obtain the same aggregate award. And thus,
				$$
				\sum_{j \in P_k} R_j(D^0) = \frac{\pi^0}{|K^0_a|}
				$$
				Within consortium $P_k$ we have two types of museums: visited ($j \in P_k$ such that $C^{(a)}_{ja}=1$) and non-visited  ($j \in P_k$ such that $C^{(a)}_{ja}=0$). \emph{Dummy} requires that any non-visited museum gets zero, while \emph{symmetry within consortia} implies that visited museum gets equal share of $\frac{\pi^0}{|K^0_a|}$. Hence, as $i \in P_k$, if $C^{(a)}_{ia}=0$ then $R_i(D^0)=0=R^{EE}_i(D^0)$. And, if $C^{(a)}_{ia}=1$, then $R_i(D^0)= \frac{\pi}{|K^0_a|}\frac{1}{|M_a \cap P^{(i)}|}=R^{EE}_i(D^0)$. Therefore, in any case,
				$$
				R_i(D^0)=R^{EE}_i(D^0)
				$$
				Now, let $D=(M,P,N,\pi,C) \in \mathcal{D}^{0}$ without any restriction on the cardinality of $N^{0}$. By \emph{composition}, it follows that, for each $i \in M$,
				\begin{align*}
					R_i (D) &= \sum_{a \in N^0} R_i \left(M,P,\{a\},\pi, C^{(a)} \right) \\
					&= \sum_{a \in N^0} R^{EE}_i \left(M,P,\{a\},\pi, C^{(a)} \right) \\
					&=R^{EE}_i (D)
				\end{align*}
				Therefore, $R$ and the egalitarian-egalitarian rule coincide in $\mathcal{D}^0$.
				\item[(iii)] Let $P_t \in P$, and let $D^t=(M,P,\{a\},\pi,C^{(a)}) \in \mathcal{D}^t$ be a problem with just one of those pass holders. If $i \in \NM(D^t)$, \emph{dummy} implies that $R_i(D^t) = 0 = R^{EE}_i(D^t)$. In application of \emph{symmetry within consortia}, all the other museums in $M \backslash \NM(D^t)$ get an equal share of $\pi^t$. And thus, if $i$ is not dummy, $R_i(D^t)=\frac{\pi^t}{|M_a|} = R^{EE}_i(D^t)$. Now, let $D=(M,P,N,\pi,C) \in \mathcal{D}^{t}$ without any restriction on the cardinality of $N^t$. \emph{Composition} guarantees that, for each $i \in M$,
				$$
				R_i (D) = \sum_{a \in N^t} R_i \left(M,P,\{a\},\pi, C^{(a)} \right) = R^{EE}_i (D)
				$$
				Thus, $R$ coincides with the egalitarian-egalitarian rule in any subclass of problems $\mathcal{D}^t$.
			\end{itemize}
			Finally, in application of Lemma \ref{lemma.4.}, $R=R^{EE}$.
		\end{proof}
	\end{theorem}
	
	The independence of the properties in Theorem \ref{thm2} is proved in the following remark.
	
	\begin{remark}
		\label{remark1}
		The axioms of Theorem \ref{thm2} are independent.
		\begin{itemize}
			\item[(a)] Let $R^1$ de defined as follows. For each $i\in M$
			$$
			R^{1}_{i}(D)=\sum_{a \in N^0_i} \frac{1}{|M_a \cap P^{(i)}|} \frac{1}{|K^0_a|} \pi^0+\frac{\sum_{a\in N^{(i)}}C_{ia}}{\sum_{j\in P^{(i)}}\sum_{a\in N^{(i)}}C_{ja}}|N^{(i)}|\pi^{(i)}+|N^{-i}|\pi^{-i}.
			$$
			The rule $R^1$ satisfies symmetry within consortia, symmetry between consortia, dummy, but not composition.
			
			\item[(b)] The rule $R^{EP}$ satisfies composition, symmetry between consortia, dummy, but not symmetry within consortia.
			
			\item[(c)] The rule $R^{PE}$ satisfies composition, symmetry within consortia, dummy, but not symmetry between consortia.
			
			\item[(d)] Let $R^2$ be defined as follows. For each $i\in M$
			$$
			R^{2}_{i}(D)=\frac{E}{|P||P^{(i)}|}.
			$$
			The rule $R^2$ satisfies composition, symmetry within consortia, symmetry between consortia, but not dummy.
		\end{itemize}
	\end{remark}

	\begin{theorem}\label{thm3}
		A rule satisfies composition, dummy, splitting-proofness of museums, splitting-proofness of consortia and consortia consistency if and only if it is the proportional-proportional rule.
		\begin{proof}
			We start by showing that the \emph{proportional-proportional rule} satisfies the axioms in the statement.
			\begin{itemize}
				\item Composition. Let $(M,P,N,\pi,C) , (M,P,N',\pi,C')  \in \mathcal{D}$ such that $N \cap N' = \emptyset$. Let $i \in M$. It follows that
				\begin{align*}
					R^{PP}_{i}(M,P,N\cup N',\pi,C\cup C') &= \sum_{a \in N^0_i\cup N'^0_i}\frac{\pi^{-i}}{\sum_{j \in M_a \cap P^{(i)}} \pi^{-j}} \frac{\pi^{(i)}}{\sum_{t \in K^0_a}\pi^t} \pi^0 \\
					&+\sum_{a \in N^{(i)}_i\cup N'^{(i)}_i}  \frac{\pi^{-i}}{\sum_{j\in M_a}\pi^{-j}}\pi^{(i)} + |N^{-i}|\pi^{-i}+|N'^{-i}|\pi^{-i} \\
					&= \sum_{a \in N^0_i} \frac{\pi^{-i}}{\sum_{j \in M_a \cap P^{(i)}} \pi^{-j}} \frac{\pi^{(i)}}{\sum_{t \in K^0_a}\pi^t} \pi^0 \\
					&+ \sum_{a \in N^{(i)}_i} \frac{\pi^{-i}}{\sum_{j\in M_a}\pi^{-j}} \pi^{(i)}+|N^{-i}|\pi^{-i} \\
					&+ \sum_{a \in N'^0_i} \frac{\pi^{-i}}{\sum_{j \in M_a \cap P^{(i)}} \pi^{-j}} \frac{\pi^{(i)}}{\sum_{t \in K^0_a}\pi^t} \pi^0 \\
					&+ \sum_{a \in N'^{(i)}_i} \frac{\pi^{-i}}{\sum_{j\in M_a}\pi^{-j}} \pi^{(i)}+|N'^{-i}|\pi^{-i}\\
					&=R^{PP}_{i}(M,P,N,\pi,C) + R^{PP}_{i}(M,P,N',\pi,C') .
				\end{align*}
				\item Dummy. Let $D \in \mathcal{D}$ and $i\in \NM(D)$. Then
				$$
				R^{PP}_{i}(D) = \sum_{a \in N^0_i} \frac{\pi^{-i}}{\sum_{j \in M_a \cap P^{(i)}} \pi^{-j}} \frac{\pi^{(i)}}{\sum_{t \in K^0_a}\pi^t} \pi^0+\sum_{a \in N^{(i)}_i} \frac{\pi^{-i}}{\sum_{j\in M_a}\pi^{-j}} \pi^{(i)}+|N^{-i}|\pi^{-i}=0.
				$$
				\item Consortia consistency. Let $D \in \mathcal{D}^0$ and $P_k \in P$,
				\begin{align*}
					\sum_{i\in P_k} R^{PP}_{i}(D) &= \sum_{i\in P_k}\sum_{a \in N^0_i} \frac{\pi^{-i}}{\sum_{j \in M_a \cap P^{(i)}} \pi^{-j}} \frac{\pi^{k}}{\sum_{t \in K^0_a}\pi^t} \pi^0 = \sum_{a \in N^0:k\in K^0_a} \frac{\pi^{k}}{\sum_{t \in K^0_a}\pi^t} \pi^0 \\
					&= R^{PP}_k\left(\{1,...,s\},\hat{P}^s,\hat{N},\hat{\pi},\hat{C}\right).
				\end{align*}
				\item Splitting-proofness of museums. Let $(M,P,N,\pi,C)  \in \mathcal{D}^0$ and $i \in M$. Consider $(M',P',N',\pi',C') \in \mathcal{D}$ as it is described in the definition of the property.\footnote{For sake of exposition and brevity, we do not replicate here the tedious definitions of the elements of $(M',P',N',\pi',C')$, which are already described in the statement of the axiom.} Then, for each $j\in M\backslash\{i\}$, if $j\notin P^{(i)}$,
				\begin{align*}
					R^{PP}_{j}(D) & =\sum_{a \in N^0_j} \frac{\pi^{-j}}{\sum_{l \in M_a \cap P^{(j)}} \pi^{-l}} \frac{\pi^{(j)}}{\sum_{t \in K^0_a}\pi^t} \pi^0+\sum_{a \in N^{(j)}_j} \frac{\pi^{-j}}{\sum_{l\in M_a}\pi^{-l}} \pi^{(j)}+|N^{-j}|\pi^{-j} \\
					&= R^{PP}_j(M',P',N',\pi',C').
				\end{align*}
				If $j\in P^{(i)}$,
				\begin{align*}
					R^{PP}_{j}(D) &= \sum_{a \in N^0_j} \frac{\pi^{-j}}{\sum_{l \in M_a \cap P^{(j)}} \pi^{-l}} \frac{\pi^{(j)}}{\sum_{t \in K^0_a}\pi^t} \pi^0+\sum_{a \in N^{(j)}_j} \frac{\pi^{-j}}{\sum_{l\in M_a}\pi^{-l}} \pi^{(j)}+|N^{-j}|\pi^{-j} \\
					&= \sum_{a \in N^0_j} \frac{\pi^{-j}}{\sum_{l \in M'_a \cap P'^{(j)}} \pi^{-l}} \frac{\pi^{(j)}}{\sum_{t \in K^0_a}\pi^t} \pi^0+\sum_{a \in N^{(j)}_j} \frac{\pi^{-j}}{\sum_{l\in M'_a}\pi^{-l}} \pi^{(j)}+|N^{-j}|\pi^{-j} \\
					&= R^{PP}_j(M',P',N',\pi',C')
				\end{align*}
				where $P'^{(j)}=\left(P^{(i)}\backslash\{i\}\right)\cup\{i_1,\dots,i_r\}$ and for each $a\in N^{\sigma}_j$ with $\sigma\in\{0,(j)\}$, $M'_a=\{l\in M':C'^{\sigma}_{la}=1\}$. Therefore, if $i\in M_a$ then $i_h\in M'_a$ for all $h\in\{1,\dots,r\}$. Since $\pi^{-i}=\sum_{h=1}^{r}\pi^{-i_h}$ we have $\sum_{l\in M_a}\pi^{-l}=\sum_{l\in M'_a}\pi^{-l}$.
				\item Splitting-proofness of consortia. Let $(M,P,N,\pi,C) \in \mathcal{D}$, and consider $(M',P',N',\pi',C') \in \mathcal{D}$ as it is set in the definition of the property. Then, for each $P_r \in P \backslash \{P_k\}$
				\begin{align*}
					\sum_{j\in P_r}R^{PP}_j(M,P,N,\pi,C)=\\
					\sum_{j\in P_r}\left(\sum_{a \in N^0_j} \frac{\pi^{-j}}{\sum_{l \in M_a \cap P^{(j)}} \pi^{-l}} \frac{\pi^{(j)}}{\sum_{t \in K^0_a}\pi^t} \pi^0+\sum_{a \in N^{(j)}_j} \frac{\pi^{-j}}{\sum_{l\in M_a}\pi^{-l}} \pi^{(j)}+|N^{-j}|\pi^{-j}\right)=\\
					\sum_{a \in N^0: r\in K^0_a} \frac{\pi^{(r)}}{\sum_{t \in K^0_a}\pi^t} \pi^0+|N^{(r)}| \pi^{(r)}+\sum_{j\in P_r}|N^{-j}|\pi^{-j}=\\
					\sum_{a \in N^0: r\in K'^0_a} \frac{\pi^{(r)}}{\sum_{t \in K^0_a}\pi^t} \pi^0+|N^{(r)}| \pi^{(r)}+\sum_{j\in P_r}|N^{-j}|\pi^{-j}=\\
					\sum_{j\in P_r}R^{PP}_j(M',P',N',\pi',C').
				\end{align*}
				where if $k\in K^0_a$ then $k^h\in K^0_a$ for all $h\in\{1,\dots,t\}$. Since $\pi^{k}=\sum_{h=1}^{t}\pi^{k^h}$ we have $\sum_{t \in K^0_a}\pi^t=\sum_{t \in K'^0_a}\pi^t$. By other way,  if $k\notin K^0_a$ then $ K^0_a= K'^0_a$.

			\end{itemize}
			
			Now, we prove the converse. Let $R$ be a rule that satisfies the properties in the statement, and let $(M,P,N,\pi,C) \in \mathcal{D}$. We divide the proof into several steps.
			\begin{itemize}
				\item[(i)] Let $i \in M$, and let $(M,P,N,\pi,C) \in \mathcal{D}^{-i}$. By Lemma \ref{lemma.3.}, $R_i(M,P,N,\pi,C)=\pi^{-i} |N^{-i}|= R^{PP}_i(M,P,N,\pi,C)$. Thus, $R$ and $R^{PP}$ coincide in any subclass of problems $\mathcal{D}^{-i}$.
				\item[(ii)] Let $D^0=(M,P,\{a\},\pi,C^{(a)}) \in \mathcal{D}^0$ be a problem with just one pass holder purchasing the general pass, and let $D^{0}_q=(\{1,\dots,s\},\hat{P}^s,\{a\},\hat{\pi},\hat{C}^{(a)})$ be the associated problem as it is indicated in the definition of consortia consistency. As $R$ satisfies \emph{dummy}, for all $k\in \{1,\dots,s\}$ such that $\hat{C}^{(a)}_{ka}=0$, we have that $R_k\left(D^{0}_q\right)=0$. Let us consider $K_a^0$ the set of consortia visited by the pass holder $a$. By definition of rule, there must exist at least one consortium $k\in K_a^0$ such that $R_k\left(D^{0}_q\right)=b>0$. Then, we split the consortium $k$ into $h$ new consortia such that $\pi^{k_i}=\pi^{-k_i}=\frac{\pi^k}{h}$ for each $i\in\{1,\dots,h\}$. Let $D'^{0}_q$ be the problem where the consortium $k$ splits into $h$ new consortia. Lemma \ref{lemma.2.} applied to $D'^{0}_q$ implies that $R_{k_i}\left(D'^{0}_q\right)=R_{k_j}\left(D'^{0}_q\right)$ for all $i,j\in\{1,\dots,h\}$. By \emph{splitting-proofness of consortia} $R_k\left(D^{0}_q\right)=\sum_{i=1}^{h}R_{k_i}\left(D'^{0}_q\right)=hR_{k_1}\left(D'^{0}_q\right)$. Therefore, $R_{k_1}\left(D'^{0}_q\right)=\frac{b}{h}$.
				
				Now let us consider another consortia $r\in K_a^0$. We can split the consortium $r$ into $\displaystyle h_r=\left\lfloor\frac{\pi^{r}}{\frac{\pi^k}{h}}\right\rfloor+1$ consortia\footnote{As $h$ can be as arbitrarily large as we need, we can always assume that $\displaystyle\left\lfloor\frac{\pi^{r}}{\frac{\pi^k}{h}}\right\rfloor \geq 1$.} such that $\pi^{r_i}=\pi^{-r_i}=\frac{\pi^k}{h}$ for each $i\in\{1,\dots, h_r-1\}$, and $\pi^{r_{h_r}}=\pi^{-r_{h_r}}=\pi^r-\displaystyle \sum_{i=1}^{h_r-1}\pi^{r_i}$. Let $D''^{0}_q$ be the problem where the consortium $r$ splits into $h_r$ new consortia. Lemma \ref{lemma.2.} applied to $D''^{0}_q$ implies that $R_{r_i}\left(D''^{0}_q\right)=R_{k_1}\left(D''^{0}_q\right)$ for each $i\in\{1,\dots,h_r-1\}$ and by \emph{splitting-proofness of consortia}
				$$R_{k_1}\left(D''^{0}_q\right)=R_{k_1}\left(D'^{0}_q\right)=\frac{b}{h}.$$
				Therefore,
				$$R_r\left(D^{0}_q\right)=R_r\left(D'^{0}_q\right)=\sum_{i=1}^{h_r-1}R_{r_i}\left(D''^{0}_q\right)+R_{r_{h_r}}\left(D''^{0}_q\right)=\displaystyle\left(h_r-1\right)\frac{b}{h}+R_{r_{h_r}}\left(D''^{0}_q\right).$$

				
				
				
				Now, by Lemma \ref{lemma.5.} we have that $R_{r_{h_r}}\left(D''^{0}_q\right)\leq R_{r_1}\left(D''^{0}_q\right)=\frac{b}{h}$. Finally,  
				$$\frac{R_k\left(D^{0}_q\right)}{R_r\left(D^{0}_q\right)}=\frac{b}{\left(h_r-1\right) \frac{b}{h}+R_{r_{h_r}}\left(D''^{0}_q\right)}\leq \frac{b}{\left\lfloor\frac{\pi^{r}}{\frac{\pi^k}{h}}\right\rfloor \frac{b}{h}},$$
				and, by other hand,
				$$\frac{R_k\left(D^{0}_q\right)}{R_r\left(D^{0}_q\right)}=\frac{b}{\left(h_r-1\right) \frac{b}{h}+R_{r_{h_r}}\left(D''^{0}_q\right)}\geq \frac{b}{\left\lfloor\frac{\pi^{r}}{\frac{\pi^k}{h}}\right\rfloor \frac{b}{h}+\frac{b}{h}}.$$
				
				Therefore, 
				$$\frac{1}{\left\lfloor\frac{\pi^{r}}{\frac{\pi^k}{h}}\right\rfloor \frac{1}{h}+\frac{1}{h}} \leq \frac{R_k\left(D^{0}_q\right)}{R_r\left(D^{0}_q\right)} \leq \frac{1}{\left\lfloor\frac{\pi^{r}}{\frac{\pi^k}{h}}\right\rfloor \frac{1}{h}}.$$
				
				Given that $\displaystyle \lim_{h \rightarrow +\infty} \left\lfloor\frac{\pi^{r}}{\frac{\pi^k}{h}}\right\rfloor \frac{1}{h}= \frac{\pi^r}{\pi^k}$ then $\displaystyle \frac{R_k\left(D^{0}_q\right)}{R_r\left(D^{0}_q\right)} = \frac{\pi^k}{\pi^r}$. Now, this can be done for any pair of consortia in $K^0_a$, therefore we have that for every $t\in\{1,\dots,s\}$
				$$
				R_t\left(D^{0}_q\right)=
				\begin{cases}
					0 & \text{if } t \text{ is dummy in } D^{0}_q \\
					\dfrac{\pi^t}{\sum_{r \in K^{0}_a}\pi^r} \pi^0 & \text{otherwise}
				\end{cases}
				$$ 
				
				Now, consider again the museum pass problem $D^{0}$. For each partition $P_t \in P$, \emph{consortia consistency} requires that
				\begin{equation}\label{eqn1}
					\sum_{i\in P_t}R_i\left(D^{0}\right)=R_t\left(D^{0}_q\right).
				\end{equation}
				
				Let us now consider a consortium $P_t\in P$ such that $M_a\cap P_t \neq \varnothing$ and, therefore, $R_t\left(D^{0}_q\right) >0$. Since $R$ satisfies \emph{dummy} for all $i\in P_t$ such that $C^{(a)}_{ia}=0$ then $R_i\left(D^{0}\right)=0$. Let us consider $M_a\cap P_t$ the set of museums in the consortium $P_t$ visited by the pass holder $a$. By \emph{consortium consistency}, it is glaringly obvious that there must exist at least a museum $i\in M_a\cap P_t$ such that $R_i\left(D^{0}\right)=c>0$. Then, we split the museum $i$ into $h$ new museums such that $\pi^{-i_l}=\frac{\pi^{-i}}{h}$ for each $l\in\{1,\dots,h\}$. Let $D'^{0}$ be the problem where the museum $i$ splits into $h$ new museums. Lemma \ref{lemma.1.} applied to $D'^{0}$ implies that $R_{i_l}\left(D'^{0}\right)=R_{i_r}\left(D'^{0}\right)$ for each pair $l,r\in\{1,\dots,h\}$. By \emph{splitting-proofness of museums}, $R_i\left(D^{0}\right)=\sum_{l=1}^{h}R_{i_l}\left(D'^{0}\right)=hR_{i_1}\left(D'^{0}\right)$. Therefore, $R_{i_1}\left(D'^{0}\right)=\frac{c}{h}$. Now let us consider another museum $j\in M_a\cap P_t$. We can split the museum $j$ into $\displaystyle h_j = \left\lfloor\frac{\pi^{-j}}{\frac{\pi^{-i}}{h}}\right\rfloor+1$ museums\footnote{As $h$ can be as arbitrarily large as we need, we can always assume that $\displaystyle\left\lfloor\frac{\pi^{-j}}{\frac{\pi^{-i}}{h}}\right\rfloor \geq 1$.} such that for each $l\in\{1,\dots, h_j-1\}$, $\pi^{-j_l}=\frac{\pi^{-i}}{h}$ and $\pi^{-j_{h_j}}=\pi^{-j}-\displaystyle \sum_{l=1}^{h_j-1}\pi^{-j_l}$. Let $D''^{0}$ be the problem where the museum $j$ splits into $h_j$ new museums. Lemma \ref{lemma.1.} applied to $D''^{0}$ implies that $R_{j_l}\left(D''^{0}\right)=R_{i_1}\left(D''^{0}\right)$ for each $l\in\{1,\dots,h_j-1\}$ and by \emph{splitting-proofness of museums}
				$$R_{i_1}\left(D''^{0}\right)=R_{i_1}\left(D'^{0}\right)=\frac{c}{h}$$
				and
				$$
				R_j\left(D^{0}\right)=R_j\left(D'^{0}\right)=\sum_{l=1}^{h_j-1}R_{j_l}\left(D''^{0}\right)+R_{j_{h_j}}\left(D''^{0}\right)=\displaystyle \left(h_j-1\right)\frac{c}{h}+R_{j_{h_j}}\left(D''^{0}\right).
				$$
				
				On one hand, by Lemma \ref{lemma.6.}, we have that $R_{j_{h_j}}\left(D''^{0}\right) \leq R_{j_1}\left(D''^{0}\right)=\frac{c}{h}$. Finally, 
				$$\frac{R_i\left(D^{0}\right)}{R_j\left(D^{0}\right)}=\frac{c}{\left(h_j-1\right)\frac{c}{h}+R_{j_{h_j}}\left(D''^{0}\right)}\leq \frac{c}{\left(h_j-1\right)\frac{c}{h}}$$
				On the other hand,
				$$\frac{R_i\left(D^{0}\right)}{R_j\left(D^{0}\right)}=\frac{c}{\left(h_j-1\right)\frac{c}{h}+R_{j_{h_j}}\left(D''^{0}\right)}\geq \frac{c}{\left(h_j-1\right)\frac{c}{h}+\frac{c}{h}}.$$
				
				Therefore, 
				$$
				\frac{1}{\left\lfloor\frac{\pi^{-j}}{\frac{\pi^{-i}}{h}}\right\rfloor \frac{1}{h}+\frac{1}{h}} \leq \frac{R_i\left(D^{0}\right)}{R_j\left(D^{0}\right)} \leq \frac{1}{\left\lfloor\frac{\pi^{-j}}{\frac{\pi^{-i}}{h}}\right\rfloor \frac{1}{h}}.
				$$
				Given that $\displaystyle \lim_{h \rightarrow +\infty} \left\lfloor\frac{\pi^{-j}}{\frac{\pi^{-i}}{h}}\right\rfloor \frac{1}{h}= \frac{\pi^{-j}}{\pi^{-i}}$ then $\displaystyle \frac{R_i\left(D^{0}\right)}{R_j\left(D^{0}\right)} = \frac{\pi^{-i}}{\pi^{-j}}$. Now, since this can be done for any pair of museums in $P_t$, together with Equation \ref{eqn1}, we have for each $i\in P_t$
				$$
				R_i\left(D^{0s}\right)=
				\begin{cases}
					0 & \text{if } i \text{ is dummy in } D^{0} \\
					\dfrac{\pi^{-i}}{\sum_{j \in M_a\cap P_t}\pi^{-j}} \frac{\pi^t}{\sum_{r \in K^{0}_a}\pi^r} & \text{otherwise}
				\end{cases}
				$$ 
				Therefore, $R\left(D^{0}\right)$ and $R^{PP}\left(D^{0}\right)$ coincide. Now, let $(M,P,N,\pi,C) \in \mathcal{D}^{0}$ without any restriction on the cardinality of $N^{0}$. By \emph{composition}, it follows that, for each $i \in M$,
				\begin{align*}
					R_i \left(M,P,N,\pi, C \right) &= \sum_{a \in N^0} R_i \left(M,P,\{a\},\pi, C^{(a)} \right) \\
					&= \sum_{a \in N^0} R^{PP}_i \left(M,P,\{a\},\pi, C^{(a)} \right) \\
					&=R^{PP}_i \left(M,P,N,\pi, C \right)
				\end{align*}
				Therefore, $R$ and $R^{PP}$ coincide in $\mathcal{D}^0$.

				\item[(iii)] Let $P_t \in P$, and let $D^t=(M,P,\{a\},\pi,C^{(a)}) \in \mathcal{D}^t$ be a problem with just one of those pass holders. Let $i \in P_t$. If $i \in \NM(D^t)$, \emph{dummy} implies that $R_i(D^t) = 0 = R^{PP}_i(D^t)$. Now, following similar arguments to those used in (ii) above, we have that
				$$
				R_i\left(D^{t}\right) = \frac{\pi^{-i}}{\sum_{j \in M_a} \pi^{-j}} \pi^t = R_i^{PP}(D^t)
				$$
				As in case (ii), by \emph{composition} we extend the coincidence between $R$ and $R^{PP}$ to the whole subclass $D^t$.

				
			\end{itemize}
			In application of Lemma \ref{lemma.4.}, $R$ and $R^{PP}$ coincide in the whole domain $\mathcal{D}$.
		\end{proof}
	\end{theorem}
	
	The independence of the properties in Theorem \ref{thm3} is proved in the following remark.
	
	\begin{remark}
		\label{remark2}
		The axioms of Theorem \ref{thm3} are independent.
		\begin{itemize}
			\item[(a)] Let $R^3$ be defined as follows. For each $i\in M$
			$$
			R^{3}_{i}(D)=\sum_{a \in N^0_i} \frac{\pi^{-i}}{\sum_{j \in M_a \cap P^{(i)}} \pi^{-j}} \frac{\pi^{(i)}}{\sum_{t \in K^0_a}\pi^t} \pi^0+ \frac{\pi^{-i}\left\{1_{\exists a\in N^{(i)}: C_{ia}=1}\right\}}{\sum_{j\in P^{(i)}}\pi^{-j}\left\{1_{\exists a\in N^{(i)}: C_{ja}=1}\right\}} |N^{(i)}|\pi^{(i)}+|N^{-i}|\pi^{-i}
			$$
			The rule $R^3$ satisfies dummy, splitting-proofness of museums, splitting-proofness of consortia, consortia consistency, but not composition.
			
			\item[(b)] The rule $R^{EP}$ satisfies composition, dummy, splitting-proofness of museums, consortia consistency, but not splitting-proofness of consortia.
			
			\item[(c)] The rule $R^{PE}$ satisfies composition, dummy, splitting-proofness of consortia, consortia consistency, but not splitting-proofness of museums.
			
			\item[(d)] Let $R^4$ be defined as follows. For each $i\in M$
			$$
			R^{4}_{i}(D)=\sum_{a \in N^0_i} \frac{\pi^{-i}}{\sum_{j \in M_a \cap P^{(i)}} \pi^{-j}} \frac{\pi^{(i)}}{\sum_{t \in K^0_a}\pi^t} \pi^0+\sum_{a \in N^{(i)}} \frac{\pi^{-i}}{\sum_{j\in P^{(i)}}\pi^{-j}} \pi^{(i)}+|N^{-i}|\pi^{-i}
			$$
			The rule $R^4$ satisfies composition, splitting-proofness of museums, splitting-proofness of consortia, consortia consistency, but not dummy.
			
			\item[(e)] Let $R^5$ be defined as follows. For all $i\in M$
			$$
			R^{5}_{i}(D)=\sum_{a \in N^0_i} \frac{\pi^{-i}}{\sum_{j \in M_a } \pi^{-j}} \pi^0+\sum_{a \in N^{(i)}_i} \frac{\pi^{-i}}{\sum_{j\in M_a}\pi^{-j}} \pi^{(i)}+|N^{-i}|\pi^{-i}
			$$
			The rule $R^5$ satisfies composition, splitting-proofness of museums, splitting-proofness of consortia, dummy, but not consortia consistency.
		\end{itemize}
	\end{remark}
	
	\begin{theorem}\label{thm4}
		A rule satisfies composition, dummy, symmetry within consortia, splitting-proofness of consortia and consortia consistency if and only if it is the proportional-egalitarian rule.
		\begin{proof}
			We start by showing that the \emph{proportional-egalitarian rule} satisfies the axioms in the statement.
			\begin{itemize}
				\item Composition. Let $(M,P,N,\pi,C) , (M,P,N',\pi,C')  \in \mathcal{D}$ such that $N \cap N' = \emptyset$. Let $i \in M$. It follows that
				\begin{align*}
					R^{PE}_{i}(M,P,N\cup N',\pi,C\cup C') &= \sum_{a \in N^0_i\cup N'^0_i} \frac{1}{|M_a \cap P^{(i)}|} \frac{\pi^{(i)}}{\sum_{t \in K^0_a}\pi^t} \pi^0 \\
					&+ \sum_{a \in N^{(i)}_i\cup N'^{(i)}_i} \frac{1}{|M_a|}\pi^{(i)}+
					|N^{-i}|\pi^{-i}+|N'^{-i}|\pi^{-i} \\
					&= \sum_{a \in N^0_i} \frac{1}{|M_a \cap P^{(i)}|} \frac{\pi^{(i)}}{\sum_{t \in K^0_a}\pi^t} \pi^0+\sum_{a \in N^{(i)}_i} \frac{1}{|M_a|}\pi^{(i)}+|N^{-i}|\pi^{-i} \\
					&+ \sum_{a \in N'^0_i} \frac{1}{|M_a \cap P^{(i)}|} \frac{\pi^{(i)}}{\sum_{t \in K^0_a}\pi^t} \pi^0+\sum_{a \in N'^{(i)}_i} \frac{1}{|M_a|}\pi^{(i)}+|N'^{-i}|\pi^{-i} \\
					&= R^{PE}_{i}(M,P,N,\pi,C) + R^{PE}_{i}(M,P,N',\pi,C') .
				\end{align*}
				\item Dummy. Let $D \in \mathcal{D}$ and $i\in \NM(D)$. Then
				$$
				R^{PE}_{i}(D) = \sum_{a \in N^0_i} \frac{1}{|M_a \cap P^{(i)}|} \frac{\pi^{(i)}}{\sum_{t \in K^0_a}\pi^t} \pi^0+\sum_{a \in N^{(i)}_i} \frac{1}{|M_a|}\pi^{(i)}+|N^{-i}|\pi^{-i}=0.
				$$
				\item Consortia consistency. Let $D \in \mathcal{D}^0$ and $P_k \in P$,
				\begin{align*}
					\sum_{i\in P_k} R^{PE}_{i}(D) &= \sum_{i\in P_kk} \sum_{a \in N^0_i} \frac{1}{|M_a \cap P^{(i)}|} \frac{\pi^{(i)}}{\sum_{t \in K^0_a}\pi^t} \pi^0 = \sum_{a \in N^0:k\in K^0_a} \frac{\pi^{k}}{\sum_{t \in K^0_a}\pi^t} \pi^0 \\
					&= R^{PE}_k\left(\{1,...,s\},\hat{P}^s,\hat{N},\hat{\pi},\hat{C}\right).
				\end{align*}
				\item Symmetry within consortia. Let $D \in \mathcal{D}$ and let $i,j\in M$ such that $P^{(i)}=P^{(j)}$ and satisfy the conditions in the definition of the property. Then 
				\begin{align*}
					R^{PE}_{i}(D) &= \sum_{a \in N^0_i} \frac{1}{|M_a \cap P^{(i)}|} \frac{\pi^{(i)}}{\sum_{t \in K^0_a}\pi^t} \pi^0+\sum_{a \in N^{(i)}_i} \frac{1}{|M_a|}\pi^{(i)}+|N^{-i}|\pi^{-i} \\
					&= \sum_{a \in N^0_j} \frac{1}{|M_a \cap P^{(j)}|} \frac{\pi^{(j)}}{\sum_{t \in K^0_a}\pi^t} \pi^0+\sum_{a \in N^{(j)}_j} \frac{1}{|M_a|}\pi^{(j)}+|N^{-j}|\pi^{-j} \\
					&= R^{PE}_{j}(D). 
				\end{align*}
				
				\item Splitting-proofness of consortia. Let $(M,P,N,\pi,C) \in \mathcal{D}$ and $P_k\in P$ where $P_k=\{i_1,\dots,i_r\}$, and consider $(M',P',N',\pi',C') \in \mathcal{D}$ as it is set in the definition of the property. Then, for each $P_r \in P \backslash \{P_k\}$
				\begin{align*}
					\sum_{j\in P_r}R^{PE}_j(M,P,N,\pi,C)=\\
					\sum_{j\in P_r}\left(\sum_{a \in N^0_j} \frac{1}{|M_a \cap P^{(j)}|} \frac{\pi^{(j)}}{\sum_{t \in K^0_a}\pi^t} \pi^0+\sum_{a \in N^{(j)}_j} \frac{1}{|M_a|}\pi^{(j)}+|N^{-j}|\pi^{-j}\right)=\\
					\sum_{a \in N^0:k\in K^0_a} \frac{\pi^{(k)}}{\sum_{t \in K^0_a}\pi^t} \pi^0+|N^{(r)}|\pi^{(r)}+\sum_{j\in P_r}|N^{-j}|\pi^{-j}=\\
					\sum_{a \in N^0:k\in K'^0_a} \frac{\pi^{(k)}}{\sum_{t \in K^0_a}\pi^t} \pi^0+|N^{(r)}|\pi^{(r)}+\sum_{j\in P_r}|N^{-j}|\pi^{-j}=\\
					\sum_{j\in P_r}R^{PE}_j(M',P',N',\pi',C').
				\end{align*}
				where if $k\in K^0_a$ then $k^h\in K^0_a$ for all $h\in\{1,\dots,t\}$. Since $\pi^{k}=\sum_{h=1}^{t}\pi^{k^h}$ we have $\sum_{t \in K^0_a}\pi^t=\sum_{t \in K'^0_a}\pi^t$. By other way,  if $k\notin K^0_a$ then $ K^0_a= K'^0_a$.
				
			\end{itemize}
			
			Now, we prove the converse. Let $R$ be a rule that satisfies the properties in the statement, and let $(M,P,N,\pi,C) \in \mathcal{D}$. We divide the proof into several steps.
			\begin{itemize}
				\item[(i)] Let $i \in M$, and let $(M,P,N,\pi,C) \in \mathcal{D}^{-i}$. By Lemma \ref{lemma.3.}, $R_i(M,P,N,\pi,C)=\pi^{-i} |N^{-i}|= R^{PE}_i(M,P,N,\pi,C)$. Thus, $R$ and $R^{EP}$ coincide in any subclass of problems $\mathcal{D}^{-i}$.
				\item[(ii)] Let $D^0=(M,P,\{a\},\pi,C^{(a)}) \in \mathcal{D}^0$ be a problem with just one pass holder purchasing the general pass, and let $D^{0s}=(\{1,\dots,s\},\hat{P}^s,\{a\},\hat{\pi},\hat{C}^{(a)})$  its associated quotient problem. Following arguments similar to those used in Case (ii) of the proof of Theorem \ref{thm3}, together with Lemma \ref{lemma.2.} and Lemma \ref{lemma.5.}, we have
				$$
				R_t\left(D^{0s}\right)=R^{PE}_t\left(D^{0s}\right)=
				\begin{cases}
					0 & \text{if } t \text{ is dummy in } D^{0s} \\
					\dfrac{\pi^t}{\sum_{r \in K^{0s}_a}\pi^r} \pi^0 & \text{otherwise}
				\end{cases}
				$$


				Now, consider again the museum pass problem $D^{0}$. For each partition $P_t \in P$, the \emph{consortia consistency} requires that
				\begin{equation}\label{eqn2}
					\sum_{i\in P_t}R_i\left(D^{0}\right)=R_t\left(D^{0s}\right)
				\end{equation}
				Within consortium $t$, any museum $i \in P_t$ is either dummy or symmetric to any other non-dummy museum. \emph{Dummy} and \emph{symmetry within consortia} together imply that all dummy museums obtain zero, while the others receive equal shares. That is, if $\alpha \in \mathbb{R}_+$ denotes that equal share, it must hold that
				$$
				\sum_{i\in P_t}R_i\left(D^{0}\right) = \alpha |M_a \cap P_t|
				$$
				If we combine the previous expression with Equation \ref{eqn2} we obtain that
				$$
				\alpha = \frac{1}{|M_a \cap P_t|} \cdot R_t\left(D^{0s}\right)
				$$
				Therefore, if $i \in M$ is dummy, then $R_i(D^0)=0=R^{PE}(D^0)$. Otherwise,
				$$
				R_i(D^0)=\frac{1}{|M_a \cap P^{(i)}|} \cdot \frac{\pi^{(i)}}{\sum_{r \in K^{0s}_a}\pi^r} \pi^0 = R^{PE}_i(D^0).
				$$
				Now, let $(M,P,N,\pi,C) \in \mathcal{D}^{0}$ without any restriction on the cardinality of $N^{0}$. By \emph{composition}, it follows that, for each $i \in M$,
				\begin{align*}
					R_i \left(M,P,N,\pi, C \right) &= \sum_{a \in N^0} R_i \left(M,P,\{a\},\pi, C^{(a)} \right) \\
					&= \sum_{a \in N^0} R^{PP}_i \left(M,P,\{a\},\pi, C^{(a)} \right) \\
					&=R^{PE}_i \left(M,P,N,\pi, C \right)
				\end{align*}
				Therefore, $R$ and $R^{PE}$ coincide in $\mathcal{D}^0$.
				\item[(iii)] Let $P_t \in P$, and let $(M,P,N,\pi,C) \in \mathcal{D}^t$. By  an argument analogous to Case (iii) in Theorem \ref{thm2}, we can obtain that, for each $i \in P_t$,
				$$
				R_i(M,P,N,\pi,C)=\sum_{a \in N^{t}_i} \frac{1}{|M_a|}\pi^{t} = R^{PE}_i(M,P,N,\pi,C)
				$$
			\end{itemize}
			In application of Lemma \ref{lemma.4.}, $R$ and $R^{PE}$ coincide in the whole domain $\mathcal{D}$.
		\end{proof}
	\end{theorem}
	
	The independence of the properties in Theorem \ref{thm4} is proved in the following remark.
	
	\begin{remark}
		\label{remark3}
		The axioms of Theorem \ref{thm4} are independent.
		\begin{itemize}
			\item[(a)] Let $R^6$ be defined as follows. For each $i\in M$
			$$
			R^{6}_{i}(D)=\sum_{a \in N^0_i} \frac{1}{|M_a \cap P^{(i)}|} \frac{\pi^{(i)}}{\sum_{t \in K^0_a}\pi^t} \pi^0+\frac{\sum_{a\in N^{(i)}}C_{ia}}{\sum_{j\in P^{(i)}}\sum_{a\in N^{(i)}}C_{ja}}|N^{(i)}|\pi^{(i)}+|N^{-i}|\pi^{-i}
			$$
			
			The rule $R^6$ satisfies dummy, symmetry within consortia, splitting-proofness of consortia and consortia consistency, but not composition.
			
			\item[(b)] The rule $R^{EE}$ satisfies composition, dummy, symmetry within consortia, consortia consistency, but not splitting-proofness of consortia.
			
			\item[(c)] The rule $R^{PP}$ satisfies composition, dummy, splitting-proofness of consortia, consortia consistency, but not symmetry within consortia.
			
			\item[(d)] Let $R^7$ be defined as follows. For each $i\in M$
			
			$$
			R^{7}_{i}(D)=\sum_{a \in N^0_i} \frac{1}{|M_a \cap P^{(i)}|} \frac{\pi^{(i)}}{\sum_{t \in K^0_a}\pi^t} \pi^0+\sum_{a \in N^{(i)}} \frac{1}{|P^{(i)}|}\pi^{(i)}+|N^{-i}|\pi^{-i}
			$$
			
			The rule $R^7$ satisfies composition, symmetry within consortia, splitting-proofness of consortia, consortia consistency, but not dummy.
			
			\item[(e)] Let $R^8$ be defined as follows. For each $i\in M$
			$$
			R^{8}_{i}(D)=\sum_{a \in N^0_i} \frac{1}{|M_a \cap P^{(i)}|} \frac{\sum_{j\in P^{(i)}}\pi^{-j}}{\sum_{t \in K^0_a}\sum_{j\in P^{t}}\pi^{-j}} \pi^0+\sum_{a \in N^{(i)}_i} \frac{\pi^{-i}}{\sum_{j\in M_a}\pi^{-j}} \pi^{(i)}+|N^{-i}|\pi^{-i}
			$$
			The rule $R^8$ satisfies satisfies composition, symmetry within consortia, splitting-proofness of consortia, dummy, but not consortia consistency.
		\end{itemize}
	\end{remark}

	\begin{theorem}\label{thm5}
		A rule satisfies composition, dummy, symmetry between consortia, and splitting-proofness of museums if and only if it is the egalitarian-proportional rule.
		\begin{proof}
			We start by showing that the \emph{egalitarian-proportional rule} satisfies the axioms in the statement.
			\begin{itemize}
				\item Composition. Let $(M,P,N,\pi,C) , (M,P,N',\pi,C')  \in \mathcal{D}$ such that $N \cap N' = \emptyset$. Let $i \in M$. It follows that
				\begin{align*}
					R^{EP}_{i}(M,P,N\cup N',\pi,C\cup C') &= \sum_{a \in N^0_i\cup N'^0_i}\frac{\pi^{-i}}{\sum_{j \in M_a \cap P^{(i)}} \pi^{-j}} \frac{1}{|K^0_a|} \pi^0 \\
					&+ \sum_{a \in N^{(i)}_i\cup N'^{(i)}_i}  \frac{\pi^{-i}}{\sum_{j\in M_a}\pi^{-j}} \pi^{(i)} + |N^{-i}|\pi^{-i}+|N'^{-i}|\pi^{-i} \\
					&= \sum_{a \in N^0_i} \frac{\pi^{-i}}{\sum_{j \in M_a \cap P^{(i)}} \pi^{-j}} \frac{1}{|K^0_a|} \pi^0+\sum_{a \in N^{(i)}_i} \frac{\pi^{-i}}{\sum_{j\in M_a}\pi^{-j}} \pi^{(i)}+|N^{-i}|\pi^{-i} \\
					&+ \sum_{a \in N'^0_i} \frac{\pi^{-i}}{\sum_{j \in M_a \cap P^{(i)}} \pi^{-j}} \frac{1}{|K^0_a|} \pi^0+\sum_{a \in N'^{(i)}_i} \frac{\pi^{-i}}{\sum_{j\in M_a}\pi^{-j}} \pi^{(i)}+|N'^{-i}|\pi^{-i} \\
					&= R^{EP}_{i}(M,P,N,\pi,C) + R^{EP}_{i}(M,P,N',\pi,C') .
				\end{align*}
				\item Dummy. Let $D \in \mathcal{D}$ and $i\in \NM(D)$. Then
				$$
				R^{EP}_{i}(D) = \sum_{a \in N^0_i} \frac{\pi^{-i}}{\sum_{j \in M_a \cap P^{(i)}} \pi^{-j}} \frac{1}{|K^0_a|} \pi^0+\sum_{a \in N^{(i)}_i} \frac{\pi^{-i}}{\sum_{j\in M_a}\pi^{-j}} \pi^{(i)}+|N^{-i}|\pi^{-i}=0.
				$$
				\item Symmetry between consortia. Let $D \in \mathcal{D}$ and let $P_r,P_t\in P$ that satisfy the conditions in the definition of the property. Then
				\begin{align*}
					\sum_{i\in P_r} R^{EP}_{i}(D) &= \sum_{i\in P_r}\sum_{a \in N^0_i} \frac{\pi^{-i}}{\sum_{j \in M_a \cap P^{(i)}} \pi^{-j}} \frac{1}{|K^0_a|} \pi^0 + \sum_{i\in P_r}\sum_{a \in N^{(i)}_i} \frac{\pi^{-i}}{\sum_{j\in M_a}\pi^{-j}} \pi^{(i)}+ \sum_{i\in P_r}|N^{-i}|\pi^{-i} \\
					&= \sum_{a \in N^0:r\in K^0_a} \frac{1}{|K^0_a|} \pi^0+\pi^r|N^r|+\sum_{j\in P_t}|N^{-j}|\pi^{-j} \\
					&= \sum_{a \in N^0:t\in K^0_a} \frac{1}{|K^0_a|} \pi^0+\pi^t|N^t|+\sum_{j\in P_t}|N^{-j}|\pi^{-j} \\
					&= \sum_{i\in P_t}\sum_{a \in N^0_i} \frac{\pi^{-i}}{\sum_{j \in M_a \cap P^{(i)}} \pi^{-j}} \frac{1}{|K^0_a|} \pi^0 + \sum_{i\in P_t}\sum_{a \in N^{(i)}_i} \frac{\pi^{-i}}{\sum_{j\in M_a}\pi^{-j}} \pi^{(i)}+ \sum_{i\in P_t}|N^{-i}|\pi^{-i} \\
					&= \sum_{j\in P_t}  R^{EP}_{j}(M,P,N,\pi,C). 
				\end{align*}
				\item Splitting-proofness of museums. Let $D \in \mathcal{D}^0$ and $i \in M$. Consider $(M',P',N',\pi',C') \in \mathcal{D}$ as it is described in the definition of the property. Then, for each $j\in M\backslash\{i\}$, if $j\notin P^{(i)}$,
				\begin{align*}
					R^{EP}_{j}(D) &= \sum_{a \in N^0_j} \frac{\pi^{-j}}{\sum_{l \in M_a \cap P^{(j)}} \pi^{-l}} \frac{1}{|K^0_a|} \pi^0+\sum_{a \in N^{(j)}_j} \frac{\pi^{-j}}{\sum_{l\in M_a}\pi^{-l}} \pi^{(j)}+|N^{-j}|\pi^{-j} \\
					&= R^{EP}_{j}(M',P',N',\pi',C').
				\end{align*}
				If $j\in P^{(i)}$,
				\begin{align*}
					R^{EP}_{j}(D) &= \sum_{a \in N^0_j} \frac{\pi^{-j}}{\sum_{l \in M_a \cap P^{(j)}} \pi^{-l}} \frac{1}{|K^0_a|} \pi^0+\sum_{a \in N^{(j)}_j} \frac{\pi^{-j}}{\sum_{l\in M_a}\pi^{-l}} \pi^{(j)}+|N^{-j}|\pi^{-j} \\
					&= \sum_{a \in N^0_j} \frac{\pi^{-j}}{\sum_{l \in M'_a \cap P'^{(j)}} \pi^{-l}} \frac{1}{|K^0_a|} \pi^0+\sum_{a \in N^{(j)}_j} \frac{\pi^{-j}}{\sum_{l\in M'_a}\pi^{-l}} \pi^{(j)}+|N^{-j}|\pi^{-j} \\
					&= R^{EP}_j(M',P',N',\pi',C')
				\end{align*}
				where $P'^{(j)}=\left(P^{(i)}\backslash\{i\}\right)\cup\{i_1,\dots,i_r\}$ and  for each $a\in N^{\sigma}_j$ with $\sigma\in\{0,(j)\}$, $M'_a=\{l\in M':C'^{\sigma}_{la}=1\}$. Therefore, if $i\in M_a$ then $i_h\in M'_a$ for all $h\in\{1,\dots,r\}$. Since $\pi^{-i}=\sum_{h=1}^{r}\pi^{-i_h}$ we have $\sum_{l\in M_a}\pi^{-l}=\sum_{l\in M'_a}\pi^{-l}$.
			\end{itemize}

			Now, we prove the converse. Let $R$ be a rule that satisfies the properties in the statement, and let $(M,P,N,\pi,C) \in \mathcal{D}$. We divide the proof into several steps.
			\begin{itemize}
				\item[(i)] Let $i \in M$, and let $(M,P,N,\pi,C) \in \mathcal{D}^{-i}$. By Lemma \ref{lemma.3.}, $R_i(M,P,N,\pi,C)=\pi^{-i} |N^{-i}|= R^{EP}_i(M,P,N,\pi,C)$. Thus, $R$ and $R^{EP}$ coincide in any subclass of problems $\mathcal{D}^{-i}$.
				\item[(ii)] Let $D^0=(M,P,\{a\},\pi,C^{(a)}) \in \mathcal{D}^0$ be a problem with just one pass holder purchasing the general pass. As in case (ii) in Theorem \ref{thm2}, \emph{dummy} and \emph{symmetry between consortia} together imply that each consortium gets an equal share of $\pi^0$, i.e. for $P_t \in P$,
				$$
				\sum_{i \in P_t} R_i(D^0) = \frac{\pi^0}{|K^0_a|}
				$$
				
				By applying a similar reasoning to that in Case (ii) of the proof of Theorem \ref{thm3}, together with Lemmas \ref{lemma.1.} and \ref{lemma.6.}, we have that, for each $i \in P_t$,
				$$
				R_i\left(D^{0}\right) = \frac{\pi^{-i}}{\sum_{j \in M_a \cap P_t} \pi^{-j}} \frac{\pi^0}{|K^0_a|} = R^{EP}(D^0)
				$$
				By \emph{composition} we can extend the previous equivalence to the whole subclass $\mathcal{D}^0$. Indeed, let $(M,P,N,\pi,C) \in \mathcal{D}^{0}$, then, for each $i \in M$, 
				\begin{align*}
					R_i \left(M,P,N,\pi, C \right) &= \sum_{a \in N^0} R_i \left(M,P,\{a\},\pi, C^{(a)} \right) \\
					&= \sum_{a \in N^0} R^{EP}_i \left(M,P,\{a\},\pi, C^{(a)} \right) \\
					&=R^{EP}_i \left(M,P,N,\pi, C \right)
				\end{align*}
				Therefore, $R$ and $R^{EP}$ coincide in $\mathcal{D}^0$.     
				\item[(iii)] Let $P_t \in P$, and let $D=(M,P,N,\pi,C^{(a)}) \in \mathcal{D}^t$. As $R$ satisfies \emph{dummy} and \emph{splitting-proofness of museums}, we can construct a reasoning similar to Case (iii) in Theorem \ref{thm3} to obtain that, for each $i \in P_t$, 
				$$
				R_i(D) = \sum_{a \in N^t} \frac{\pi^{-i}}{\sum_{j\in M_a}\pi^{-j}} \pi^t  = R_i^{EP}(D)
				$$
			\end{itemize}
			In application of Lemma \ref{lemma.4.}, $R$ and $R^{EP}$ coincide in the whole domain $\mathcal{D}$.
		\end{proof}
	\end{theorem}
	
	The independence of the properties in Theorem \ref{thm5} is proved in the following remark.
	
	\begin{remark}
		\label{remark4}
		The axioms of Theorem \ref{thm5} are independent.
		\begin{itemize}
			\item[(a)] Let $R^9$ be defined as follows. For each $i\in M$
			$$
			R^{9}_{i}(D)=\sum_{a \in N^0_i} \frac{\pi^{-i}}{\sum_{j \in M_a \cap P^{(i)}} \pi^{-j}} \frac{1}{|K^0_a|} \pi^0+ \frac{\pi^{-i}\left\{1_{C_{ia}=1:a\in N^{(i)}}\right\}}{\sum_{j\in P^{(i)}}\pi^{-j}\left\{1_{C_{ja}=1:a\in N^{(i)}}\right\}} |N^{(i)}|\pi^{(i)}+|N^{-i}|\pi^{-i}
			$$
			
			The rule $R^9$ satisfies dummy, symmetry between consortia, and splitting-proofness of museums, but not composition.
			
			\item[(b)] The rule $R^{EE}$ satisfies composition, dummy, symmetry between consortia, but not splitting-proofness of museums.
			
			\item[(c)] The rule $R^{PP}$ satisfies composition, dummy, splitting-proofness of museums, but not symmetry between consortia.
			
			\item[(d)] Let $R^{10}$ be define as follows. For each$i\in M$
			$$
			R^{10}_{i}(D)=\sum_{a \in N^0} \frac{\pi^{-i}}{\sum_{j \in P^{(i)}} \pi^{-j}} \frac{1}{|K^0_a|} \pi^0+ \sum_{a \in N^{(i)}_i} \frac{\pi^{-i}}{\sum_{j\in M_a}\pi^{-j}} \pi^{(i)} +|N^{-i}|\pi^{-i}
			$$
			
			The rule $R^{10}$ satisfies satisfies composition, splitting-proofness of museums, symmetry between consortia, but not dummy.

		\end{itemize}
	\end{remark}

	\section{Final remarks}
	
	We have extended the fundamental museum pass problem by incorporating a more sophisticated market structure, allowing museums to be organized into multiple programs or consortia. Within this framework, we have proposed four different rules to distribute the overall revenue that is obtained from selling all the passes. These rules integrate the principles of egalitarianism and proportionality (relative to either consortium prices or individual prices) through a two-stage mechanism. Initially, revenue is distributed among consortia, followed by distribution among museums within each consortium. In accordance with this approach, four distinct rules are defined: the egalitarian-egalitarian, the proportional-proportional, the egalitarian-proportional, and the proportional-egalitarian rules. We employ an axiomatic methodology to establish the normative foundations of these allocation methods. The axioms examined in this paper are categorized into two groups: fairness and stability. Appropriate combinations of these principles characterize the four rules. In particular, the egalitarian-egalitarian rule is characterized by composition, dummy, symmetry within consortia, and symmetry between consortia (Theorem \ref{thm2}). In Theorem \ref{thm3} we show that replacing both symmetry requirements with splitting-proofness of museums, splitting-proofness of consortia, and consortia consistency univocally identifies the proportional-proportional rule. Furthermore, we find that integrating fairness and non-manipulability requirements is not only feasible but also leads to precise revenue distribution methods. Thus, composition, dummy, symmetry within consortia, splitting-proofness of consortia and consortia consistency characterize the proportional-egalitarian rule (Theorem \ref{thm4}). Dually, composition, dummy, symmetry between consortia, and splitting-proofness of museums characterize the egalitarian-proportional rule (Theorem \ref{thm5}).
	
	It is worth noting that, although all four rules satisfy consortia consistency, this axiom is only essential for the characterizations of the proportional-proportional and proportional-egalitarian rules (Theorems \ref{thm3} and \ref{thm4}). While it is possible to identify rules that satisfy the properties other than consortia consistency in Theorems \ref{thm3} and \ref{thm4}, the resulting families lack a clear formulation. Moreover, beyond the technical details, their practical interpretation remains ambiguous.
	
	In addition to the characterization of the egalitarian-egalitarian rule, this method finds further justification through its connection to cooperative games. Several authors have proposed different solution concepts for cooperative games with a priori unions (e.g. \cite{AlonsoMeijide20} and \cite{LorenzoFreire16}). Among those, the \textit{Owen value} (\cite{Owen77}) stands out as the most prominent one. Interestingly, the egalitarian-egalitarian rules coincides with the Owen value of the game $(M,v,P)$ where $M$ is the set of museums, $P$ is the set of consortia, and $v$ is such that, for each $S \subseteq M$, 
	$$
	v(S)=\pi^0|\{a\in N^0; M^0_a\subseteq S\}|+\sum_{k=1}^{s}\pi^k|\{a\in N^k; M^k_a\subseteq S\}|+\sum_{k=-m}^{-1}\pi^k|\{a\in N^{k}; M^{k}_a\subseteq S\}|
	$$
	
	The models proposed by \cite{Ginsburgh01} and \cite{Bergantinos15} primarily focus on the distribution of revenues generated exclusively from general passes, excluding the existence of consortia. In our terminology, their rules are only applicable within the subdomain $\mathcal{D}^0$. Nevertheless, there is a relationship between the rules we characterize in this paper and those proposed by the aforementioned authors, provided that appropriate assumptions are made and the analysis is confined to the subclass of problems $\mathcal{D}^0$. The egalitarian-egalitarian rule coincides with the \emph{Shapley rule} as discussed in these two papers when $P=\{\{1\}, \dots, \{m\}\}$. The egalitarian-proportional rule also coincides with the Shapley rule under the same condition, and assuming that the prices of the consortia passes are uniform. Still within $\mathcal{D}^0$, \cite{Bergantinos15} introduces the \emph{$p$-Shapley rule}, which is a particular case of the proportional-egalitarian rule (and the proportional-proportional rule) when $P=\{\{1\}, \dots, \{m\}\}$ (and, in addition, the prices of all consortia passes are equal). To some extent, our proposals provide a coherent extension of the rules established in these prior works.  
	
	Finally, we acknowledge our model may have several potential extensions. Perhaps one of the most straightforward relates to the assumption that the set of consortia must be a partition of the set of museums. This premise implies that a museum can only belong to one consortium, aside from participating in the general pass. We believe that, provided the global and individual passes remain feasible, this assumption could be relaxed to incorporate more complex market structures. The mechanisms behind the proposed rules could be adapted to this new framework. The core principles of fairness and stability that support the axioms would remain intact, although their formal statements and definitions would need to be reformulated. However, we caution that this could significantly the complexity of the notation in the definitions and proofs, potentially making the paper more challenging to read without introducing genuinely innovative or distinguishing ideas. 
	
	\section*{Statements and Declarations}
	
	\textbf{Conflicts of interest} The authors state that there is no conflict of interest.
	
	\newpage
	
	\section*{Appendix A. Lemmas}
	
	\begin{lemma}\label{lemma.4.}
		Let $R$ and $R'$ be two rules that satisfy composition. If $R$ and $R'$ coincide in all the subclasses of problems $\mathcal{D}^{-m}, \ldots, \mathcal{D}^{-1}, \mathcal{D}^{o}, \mathcal{D}^{1}, \ldots, \mathcal{D}^{s}$, then they also coincide in the general domain of problems $\mathcal{D}$.
		\begin{proof}
			Let $(M,P,N,\pi,C) \in \mathcal{D}$, and let $i \in M$. In application of \emph{composition},
			\begin{align*}
				R_i \left(M,P,N,\pi, C \right) &= \sum_{i=1}^m R_i \left(M,P,\overline{N}^{-i},\pi, \overline{C}^{-i} \right) + R_i \left(M,P,\overline{N}^{0},\pi, \overline{C}^{0} \right) \\
				&+ \sum_{t=1}^s R_i \left(M,P,\overline{N}^{t},\pi, \overline{C}^{t} \right),
			\end{align*}
			where each $\overline{N}^\sigma$ with $\sigma\in\{-i,0,t\}$ is the set of pass holders that only purchase (individual, global or consortium) passes to access $\sigma$. Notice that $\left(M,P,\overline{N}^{-i},\pi, \overline{C}^{-i}  \right) \in \mathcal{D}^{-i}$, $\left(M,P,\overline{N}^{0},\pi, \overline{C}^{0} \right) \in \mathcal{D}^0$, and $\left( M,P,\overline{N}^{t},\pi, \overline{C}^{t} \right) \in \mathcal{D}^t$. By assumption, both $R$ and $R'$ coincide in these subdomains. Therefore, 
			\begin{align*}
				R_i \left(M,P,N,\pi, C \right) &= \sum_{i=1}^m R'_i \left(M,P,\overline{N}^{-i},\pi, \overline{C}^{-i} \right) + R'_i \left(M,P,\overline{N}^{0},\pi, \overline{C}^{0} \right) +\sum_{t=1}^s R'_i \left(M,P,\overline{N}^{t},\pi, \overline{C}^{t} \right) \\
				&= R'_i \left(M,P,N,\pi, C \right)
			\end{align*}
		\end{proof}
	\end{lemma}
	
	\begin{lemma}\label{lemma.3.}
		If a rule $R$ satisfies dummy, then, for each $i \in M$ and each  $(M,P,N,\pi,C) \in \mathcal{D}^{-i}$,
		$$
		R_i (M,P,N,\pi,C) = \pi^{-i} |N^{-i}|
		$$
		\begin{proof}
			Let $i \in M$, and let $(M,P,N,\pi,C) \in \mathcal{D}^{-i}$. Notice that, any $j \in M\backslash \{i\}$ is dummy in this problem. The \emph{dummy} principle requires that 
			$R_j(M,P,N,\pi,C) = 0$. By definition of rule, it must then occur that $R_i(M,P,N,\pi,C)=\pi^{-i} |N^{-i}|$.
		\end{proof}
	\end{lemma}
	
	\begin{lemma}\label{lemma.1.}
		Let $(M,P,\{a\},\pi,C^{(a)}) \in \mathcal{D}^0$ be a problem with just one pass holder purchasing the general pass, and let $R$ be a rule that satisfies splitting-proofness of museums and dummy. If $\{i,j\} \subseteq M$ are  such that $P^{(i)}=P^{(j)}$, $C^{(a)}_{ia}=C^{(a)}_{ja}$ and $\pi^{-i}=\pi^{-j}$, then 
		$$
		R_i\left(M,P,\{a\},P,\pi, C^{(a)} \right)=R_j\left(M,P,\{a\},P,\pi, C^{(a)} \right).
		$$
		\begin{proof}
			Consider the problem $D=\left(M,P,\{a\},\pi, C^{(a)} \right) \in \mathcal{D}^0$ and two museums $\{i,j\} \subseteq N$ that fulfills the requirements in the statement, with $P_k=P^{(i)}=P^{(j)}$. If $C^{(a)}_{ia}=C^{(a)}_{ja}=0$, \emph{dummy} implies that $R_i(D)=R_j(D)=0$. If, instead, $C^{(a)}_{ia}=C^{(a)}_{ja}=1$, we can split museum $i$ into $\{i_1,i_2\}$ obtaining the problem $D^i=\left(M',P',\{a\},\pi', C'^{(a)} \right)$ where:
			\begin{itemize}
				\item $M'=\left(M\backslash\{i\}\right)\cup\{i_1,i_2\}$.
				\item $P'=\left(P\backslash P_k\right)\cup P'_k$ where $P'_k=\left(P_k\backslash\{i\}\right)\cup\{i_1,i_2\}$.
				\item $\pi'=\left(\pi\backslash\pi^{-i}\right)\cup\{\pi'^{-i_1},\pi'^{-i_2}\}$ where $\pi'^{-i_1}=\pi'^{-i_2}=\frac{\pi^{-i}}{2}$.
				\item $C'^{(a)}$ is such that $C'^{(a)}_{i_1a}=C'^{(a)}_{i_2a}=C^{(a)}_{ia}$, and $C'^{(a)}_{ja}=C^{(a)}_{ja}$  for any $j\in M\backslash\{i\}$
			\end{itemize}
			\emph{Splitting-proofness of museums} requires that  
			$$
			R_i(D)=R_{i_1}(D^i)+R_{i_2}(D^i).
			$$ 
			Similarly, we can split museum $j$ into $\{i_1,i_2\}$ obtaining the problem $D^j=\left(\overline{M},\overline{P},\{a\},\overline{\pi}, \overline{C}^{(a)} \right)$. Again, by \emph{splitting-proofness of museums},
			$$
			R_j(D)=R_{i_1}(D^j)+R_{i_2}(D^j).
			$$ 
			Now, in the problem $D^i$ we again split the museum $j$ into $\{j_1,j_2\}$ obtaining the problem $D^{ij}=\left(\overline{M}',\overline{P}',\{a\},\overline{\pi}', \overline{C}'^{(a)} \right)$ where:
			\begin{itemize}
				\item $\overline{M}'=\left(M'\backslash\{j\}\right)\cup\{j_1,j_2\}$.
				\item $\overline{P}'=\left(P'\backslash P'_k\right)\cup \overline{P}'_k$ where $\overline{P}'_k=\left(P'_k\backslash\{j\}\right)\cup\{j_1,j_2\}$.
				\item $\overline{\pi}'=\left(\pi'\backslash\pi'^{-j}\right)\cup\{\overline{\pi}'^{-j_1},\overline{\pi}'^{-j_2}\}$ where $\overline{\pi}'^{-j_1}=\overline{\pi}'^{-j_2}=\frac{\pi^{-j}}{2}$.
				\item $\overline{C}'^{(a)}$ is such that $\overline{C}'^{(a)}_{j_1a}=C''^{(a)}_{j_2a}=C'^{(a)}_{ja}$, and $\overline{C}'^{(a)}_{la}=C'^{(a)}_{la}$ for any $l\in M'\backslash\{j\}$
			\end{itemize}
			In application of \emph{splitting-proofness of museums}, $R_{i_1}(D^i)=R_{i_1}(D^{ij})$ and $R_{i_2}(D^i)=R_{i_2}(D^{ij})$. Then,
			$$
			R_i(D)=R_{i_1}(D^i)+R_{i_2}(D^i)=R_{i_1}(D^{ij})+R_{i_2}(D^{ij})
			$$
			Using an analogous argument for the problem $D^j$, we split museum $i$ into $\{j_1,j_2\}$ and obtain that 
			$$
			R_j(D)=R_{i_1}(D^j)+R_{i_2}(D^j)=R_{i_1}(D^{ji})+R_{i_2}(D^{ji})
			$$
			As $D^{ij}=D^{ji}$, we have that $R_{i_1}(D^{ij})=R_{i_1}(D^{ji})$ and $R_{i_2}(D^{ij})=R_{i_2}(D^{ji})$. And therefore, $R_i(D)=R_j(D)$.
		\end{proof}
	\end{lemma}
	
	\begin{lemma}\label{lemma.2.}
		Let $(M,P,\{a\},\pi,C^{(a)}) \in \mathcal{D}^0$ be a problem with just one pass holder purchasing the general pass, and let $R$ be a rule that satisfies splitting-proofness of consortia, dummy and consortia consistency. If $P_k,P_r\in P$ are such that $\sum_{i\in P_k}C^{(a)}_{ia}=0\Leftrightarrow \sum_{j\in P_r}C^{(a)}_{ja}=0$ and $\pi^{k}=\pi^{r}$, then 
		$$
		\sum_{i\in P_k}R_i\left(M,P,\{a\},P,\pi, C^{(a)} \right)=\sum_{j\in P_r}R_j\left(M,P,\{a\},P,\pi, C^{(a)} \right).
		$$
		\begin{proof}
			Consider the problem $D=\left(M,P,\{a\},\pi, C^{(a)} \right) \in \mathcal{D}^0$ and its associated reduced problem $D^s=\left(\{1,\dots,s\},\hat{P}^s,\{a\},\hat{\pi}, \hat{C}^{(a)} \right) \in \mathcal{D}^0$. Let $\{k,r \} \subseteq \{1,\dots,s\}$ be a pair of consortia such that $\hat{C}^{(a)}_{ka}=\hat{C}^{(a)}_{ra}$ and $\pi^{k}=\pi^{r}$. Reasoning as in Lemma \ref{lemma.1.} with \emph{splitting-proofness of consortia} instead of \emph{splitting-proofness of museums}, we obtain that $R_k(D^s)=R_r(D^s)$. Finally, by \emph{consortia consistency}, we conclude that 
			$$
			\sum_{i\in P_k}R_i\left(M,P,\{a\},P,\pi, C^{(a)} \right)=R_k\left(D^s \right)=R_r\left(D^s\right)=\sum_{j\in P_r}R_j\left(M,P,\{a\},P,\pi, C^{(a)} \right).
			$$
		\end{proof}
	\end{lemma}
	
	\begin{lemma}\label{lemma.5.}
		Let $(M,P,\{a\},\pi,C^{(a)}) \in \mathcal{D}^{0}$ be a problem with just one pass holder purchasing the general pass, and let $R$ be a rule that satisfies splitting-proofness of consortia. If $P_k,P_r\in P$ are such that $k,r\in K^0_a$ and $\pi^{k}\geq \pi^{r}$, then 
		$$
		R_k\left(\{1,...,s\},\hat{P}^s,\hat{N},\hat{\pi},\hat{C}\right)\geq R_r\left(\{1,...,s\},\hat{P}^s,\hat{N},\hat{\pi},\hat{C}\right).
		$$
		where $\left(\{1,...,s\},\hat{P}^s,\hat{N},\hat{\pi},\hat{C}\right)\in \mathcal{D}^0$ is the corresponding reduced problem.
	\end{lemma}
	\begin{proof}
		Let $(M,P,\{a\},\pi,C^{(a)}) \in \mathcal{D}^{0}$ be a problem with just one pass holder purchasing the general pass. Consider its reduced problem $D=\left(\{1,...,s\},\hat{P}^s,\hat{N},\hat{\pi},\hat{C}\right)$. Let $k,r\in K^0_a$ with $\pi^{k}\geq \pi^{r}$. Let $D'$ be the problem in which consortium $k$ splits into two new consortia $k'$ and $k''$ such that 
		$$
		\pi^{k'}=\pi^{-k'}=\pi^r \text{ and } \pi^{k''}=\pi^{-k''}=\pi^{k}-\pi^{r}
		$$
		As $R$ satisfies \emph{splitting-proofness of consortia}
		$$
		R_{k'}\left(D'\right)\leq R_{k'}\left(D'\right)+R_{k''}\left(D'\right)=R_{k}\left(D\right).
		$$
		By Lemma \ref{lemma.2.}
		$$
		R_{k'}\left(D'\right)=R_{r}\left(D'\right)=R_{r}\left(D\right).
		$$
		Therefore, $R_{r}\left(D\right)\leq R_{k}\left(D\right)$.
	\end{proof}

	\begin{lemma}\label{lemma.6.}
		Let $(M,P,\{a\},\pi,C^{(a)}) \in \mathcal{D}^{\sigma}$ be a problem with just one pass holder where $\sigma\in\{0,k\}$ with $k\in\{1,\dots,s\}$, and let $R$ be a rule that satisfies splitting-proofness of museums. If $i,j\in P^{k}$ are such that $C^{(a)}_{ia}=C^{(a)}_{ja}$ and $\pi^{-i}\geq \pi^{-j}$, then 
		$$
		R_i\left(M,P,\{a\},P,\pi, C^{(a)} \right)\geq R_j\left(M,P,\{a\},P,\pi, C^{(a)} \right).
		$$
	\end{lemma}
	\begin{proof}
		Let $D=(M,P,\{a\},\pi,C^{(a)}) \in \mathcal{D}^{\sigma}$ be a problem with just one pass holder where $\sigma\in\{0,k\}$ with $k\in\{1,\dots,s\}$. Consider $i,j\in P^{k}$ such that $C^{(a)}_{ia}=C^{(a)}_{ja}$ and $\pi^{-i}\geq \pi^{-j}$. Let $D'$ be the problem where museum $i$ splits into two new museums $i'$ and $i''$ such that
		$$
		\pi^{-i'}=\pi^{-j} \text{ and } \pi^{-i''}=\pi^{-i}-\pi^{-j}.
		$$
		As $R$ satisfies \emph{splitting-proofness of museums}
		$$
		R_{i'}\left(D'\right)\leq R_{i'}\left(D'\right)+R_{i''}\left(D'\right)=R_{i}\left(D\right).
		$$
		By Lemma \ref{lemma.1.}
		$$
		R_{i'}\left(D'\right)=R_{j}\left(D'\right)=R_{j}\left(D\right).
		$$
		Therefore, $R_{j}\left(D\right)\leq R_{i}\left(D\right)$.
	\end{proof}

	\newpage
	

\end{document}